\documentclass[a4paper,11pt]{article}
\usepackage{geometry}               
\usepackage{epstopdf}
\usepackage{amsthm}
\usepackage{amsfonts}
\usepackage{mathrsfs}
\usepackage{MnSymbol}
\usepackage{hyperref}
\usepackage{color}
\usepackage{tikz}
\usepackage{graphicx}
\usepackage[affil-it]{authblk}
\usepackage[applemac]{inputenc} 
\usepackage[toc,page]{appendix}
\newtheorem{theorem}{Theorem}[section]
\newtheorem{lemma}[theorem]{Lemma}
\newtheorem{proposition}[theorem]{Proposition}

\newtheorem{remark}[theorem]{Remark}

\newtheorem*{theorem*}{Theorem}
\newtheorem*{prop*}{Proposition}

\newtheoremstyle{named}{}{}{\itshape}{}{\bfseries}{.}{.5em}{\thmnote{#3}#1}
\theoremstyle{named}

\newcommand{\bea}{\begin{eqnarray}}
\newcommand{\eea}{\end{eqnarray}}
\def\beaa{\begin{eqnarray*}}
\def\eeaa{\end{eqnarray*}}

\newcommand{\MM}{\mathcal{M}}
\def\DD{{\mathcal D}}
\def\TT{{\mathcal T}}
\def\PP{\mathcal{P}}

\def\D{{\bf D}}

\def\R{{\bf R}}

\def\g{{\bf g}}

\def\CCC{{\Bbb C}}

\def\RR{\mathcal{R}}
\def\EE{\mathcal{E}}

\def\a{{\alpha}}

\def\b{{\beta}}

\def\De{\Delta}

\def\la{\lambda}

\def\vphi{\varphi}

\def\th{\theta}

\def\nab{\nabla}

\def\etab{{\underline \eta}}

\def\rhod{\,\dual\rho}

\def\qf{\frak{q}}
\def\pf{\mathfrak{p}}

\def\Ffr{\mathfrak{F}}
\def\Bfr{\mathfrak{B}}
\def\Xfr{\mathfrak{X}}
\def\sk{\mathfrak{s}}

\def\Mor{{\mbox{Mor}}}

\def\pr{{\partial}}
\def\les{\lesssim}

\def\c{\cdot}
\def\dual{{\,\,^*}}
\def\div{{\mbox div\,}}

\def\Hb{\,\underline{H}}
\def\Ab{\underline{A}}

\def\hot{\widehat{\otimes}}

\def\squared{\dot{\square}}
\def\lab{\label}

\def\DDov{\ov{\DD}}

\def\nn{\nonumber}
\def\ov{\overline}

\def\DDs{ \, \DD \hspace{-2.4pt}\dual    \mkern-20mu /}

\def\DDs{ \, \DD \hspace{-2.4pt}\dual    \mkern-20mu /}
\def\DDd{ \, \DD \hspace{-2.4pt}    \mkern-8mu /}

\def\Kh{\,^{(h)}K}

      \def\ntrap{trap\mkern-18 mu\big/\,}
          
\def\Mntrap{{\MM_{\ntrap}}}

\def\NN{\mathcal{N}}

\def\QQ{\mathcal{Q}}
\def\LL{\mathcal{L}}

\def\EE{\mathcal{E}}
\def\AA{\mathcal{A}}
\def\VV{\mathcal{V}}

\def\FF{\mathcal{F}}
\def\UU{\mathcal{U}}

\def\piX{\, ^{(X)}\pi}

\def\That{{\widehat{T}}}
\def\Si{\Sigma}

\def\Up{\Upsilon}

\def\g{{\bf g}}

\def\Rdot{\dot{\R}}

\def\Db{{\dot{\D}}}

\usepackage{setspace}

\newcommand{\mathsym}[1]{{}}
\newcommand{\unicode}[1]{{}}

\allowdisplaybreaks

\begin{document}

 \title{Boundedness and Decay for the Teukolsky System \\ in Kerr-Newman Spacetime II: The Case $|a| \ll M$, $|Q| <M$ \\ in Axial Symmetry}
 
 \author[1]{Elena Giorgi\footnote{elena.giorgi@columbia.edu}}
 
 \author[2]{{Jingbo Wan\footnote{jingbowan@math.columbia.edu}}}
 
\affil[1,2]{\small Department of Mathematics, Columbia University}

\maketitle

\begin{abstract}
We establish boundedness and polynomial decay results for the Teukolsky system in the exterior spacetime of very slowly rotating and strongly charged sub-extremal Kerr-Newman black holes, with a focus on axially symmetric solutions. The key step in achieving these results is deriving a physical-space Morawetz estimate for the associated generalized Regge-Wheeler system, without relying on spherical harmonic decomposition.
\end{abstract}

{\footnotesize

\setcounter{tocdepth}{2}
\tableofcontents}

\section{Introduction}

In recent years, great attention has been given to stability problems for the Einstein equation and significant progress has been made in the case of the Kerr spacetime through the analysis of the Teukolsky equation, particularly in understanding the decay properties of solutions \cite{Ma20a}\cite{DHR19a}\cite{SC20}\cite{SC23}\cite{GKS}. Here we focus on the less-studied case of the Teukolsky-like coupled system in Kerr-Newman spacetimes, solution to the Einstein-Maxwell equation. 

In the case of coupled gravitational and electromagnetic perturbations of a charged black hole, the intricate nature of the coupling
creates major difficulties in the study of the resulting equations, whose identification and analysis have long been an unresolved issue, even for mode stability results\footnote{For an approach to the linear stability of Kerr-Newman without the use of the Teukolsky system see \cite{Lili}.} \cite{Chandra}. Nevertheless,
it is possible to identify gauge invariant quantities satisfying a coupled system of Teukolsky-like equations which fully describe the coupled perturbations of Kerr-Newman black hole, as done by the first named author in \cite{Giorgi7}. For more details on the history of perturbations of charged black holes we direct the readers to the introduction of \cite{Giorgi9}.

This paper is Part II of a series aimed to obtain boundedness and decay for the Teukolsky system  derived in \cite{Giorgi7} for linearized gravitational and electromagnetic perturbations of slowly rotating Kerr-Newman black holes. 
The purpose of the present paper is to obtain decay estimates for the Teukolsky variables for axially symmetric perturbations of very slowly rotating and strongly charged black holes, i.e. in the case of $0< |a| \ll M$, $|Q| <M$.

A rough version of our main result is the following (see already Theorem \ref{theorem:unconditional-result-final} for the precise statement):

 \begin{theorem}[Main Theorem, rough version]
\lab{main-thm-intro-1}
Let $g_{M, a, Q}$ be a Kerr-Newman metric with $0< |a| \ll M$, $0<|Q|<M$, so that $a^2+Q^2<M^2$. Axisymmetric solutions to the Teukolsky system in Kerr-Newman spacetime $g_{M, a, Q}$ arising from regular initial data remain uniformly bounded in the exterior region and satisfy pointwise decay estimates.
\end{theorem}

 In Part I of this series \cite{Giorgi9}, the first named author has obtained an analogue of our main theorem for very slowly rotating and very weakly charged black holes, i.e. in the case of $0< |a|, |Q| \ll M$, with no symmetry assumption. Finally, in Part III of this series, we shall extend the result to $0< |a| \ll M$, $|Q| <M$ with no symmetry assumptions.

As in the case of scalar linear wave equations, axially symmetric solutions present two major simplifications: superradiance is effectively absent and the trapping region collapses to a physical-space hypersurface $\{r=r_{P}\}$, where $r_P$ is the largest root of the trapping polynomial $\TT= r^3-3Mr^2 + ( a^2+2Q^2)r+Ma^2$. For these reasons, the standard energy conservation law can be obtained independently of the integrated local energy decay (Morawetz) estimates, which are degenerate at the trapping radius $\{r=r_P\}$ (see \cite{DR10}\cite{Stogin17} for results for $|a|<M$ in axial symmetry). 
    On the other hand, for $a \neq 0$ outside axial symmetry the trapping region is an open region of spacetime, and this causes the energy estimates to be conditional on the control of the spacetime integrated local energy norms, making even the issue of boundedness of the solutions highly non trivial \cite{DR11b}. However, those spacetime Morawetz norms can be controlled in physical-space by higher order norms as in Andersson-Blue's method \cite{AB15}, which makes it possible to obtain  physical-space estimates for $|a|\ll M$.
 Outside axial symmetry, frequency decompositions are used in the known result for $|a|<M$   \cite{DRSR}, see also \cite{Tataru}. 
  
For the case of Teukolsky-like equations, as in previous works in vacuum \cite{Ma20}\cite{Ma20a}\cite{DHR19a}\cite{SC20}\cite{GKS}\cite{SC23}, our proof relies on the use of derived quantities from the Teukolsky variables, through a version of the Chandrasekhar transformation, which satisfy an (improved) system called generalized Regge-Wheeler system (gRW), also coupled to the original Teukolsky variables. Due to the similarity between the (generalized) Regge-Wheeler equation and the wave equation, the strategy of the proof is to obtain energy and Morawetz estimates for the gRW system, which also imply bounds on the Teukolsky variables. Unlike in Reissner-Nordstr\"om or Kerr, the gRW system in Kerr-Newman involves coupling of the derived quantities through elliptic operators which do not commute with the decomposition in spherical harmonics, so we are presented with 
the additional difficulty of having to analyze the system in physical-space\footnote{For the connection between this issue and the failure of mode stability see the introduction of \cite{Giorgi7}.}.  In particular, Part I of this series \cite{Giorgi9} makes use of an adaptation of the physical-space method by Andersson-Blue \cite{AB15}, which does not require mode decomposition but needs higher order energy norms and whose applicability to the charged case was laid out in \cite{Giorgi8}.

The coupled gRW system was first analyzed in Reissner-Nordstr\"om ($a=0$) for $|Q|\ll M$ \cite{Giorgi5}\cite{Giorgi4} and then for $|Q|<M$ \cite{Giorgi7a}. Crucial to extending the results to the full subextremal range was the uncovering of a symmetric structure in the elliptic operators appearing as coupling, which allowed for a definition of a modified energy shown to be coercive for $|Q|<M$. Observe that, even though not necessary, the proof in \cite{Giorgi7a} relies on the decomposition in spherical harmonics, which makes it difficult to extend it to the case of Kerr-Newman for the reasons outlined above. We also point out that for $a=0$ the Regge-Wheeler system completely decouples from the Teukolsky variables. For the analysis of the system in the extremal case $|Q|=M$, also making use of decomposition in spherical harmonics, see \cite{Apetroaie23}. 

In the case of Kerr-Newman, the gRW system was analyzed for $|a|, |Q|\ll M$ in Part I of this series \cite{Giorgi9}, without relying on decomposition in spherical harmonics. In order to obtain (conditional) energy estimates, a modified version of the symmetric structure of the coupling first uncovered in Reissner-Nordstr\"om \cite{Giorgi7a} was shown to extend to the case of Kerr-Newman, allowing to obtain energy estimates for $|a| \ll M$ and $|Q|<M$ in \cite{Giorgi9}. Notice that in \cite{Giorgi9} the restriction to small charge $|Q|\ll M$ is not only due to the choice of functions in the Morawetz multiplier, but is also more critically needed to obtained commuted energy estimates with the higher order operators due to an interesting interaction between $a\neq 0$ and $Q\neq 0$ in the commuted equations, see Remark 4.3 in \cite{Giorgi9}.

 In this paper, by restricting to axial symmetry, we avoid the issue of the commuted energy estimates necessary in Andersson-Blue's method \cite{AB15}, and make use of the fact that the integrated local energy norms of the solutions can be controlled in terms of the standard energy. In particular, the conditional energy estimates obtained in \cite{Giorgi9} for $|a| \ll M$ and $|Q|<M$ can be straightforwardly applied in this case, see Proposition  \ref{prop:energy-estimates-conditional}. 
For this reason, the main result of this paper is the derivation of the Morawetz-type estimate for axially symmetric solutions to the gRW system, which are valid for $|Q|<M$ and whose proof does not rely on the decomposition in spherical harmonics.

\subsection*{The choice of multiplier}

The study of Morawetz-type estimates has been crucial in analyzing the stability of various spacetimes, beginning with Morawetz's seminal work on the wave equation in Minkowski spacetime \cite{Morawetz61}. In the case of the linear wave equation in black hole spacetimes, the integrated local energy decay estimates need to degenerate at trapping \cite{Sbierski} because of the presence of bounded null geodesics in the spacetime.

In this paper, we derive the Morawetz estimates for the axially symmetric gRW equation through the use of a  novel choice of multiplier, which depends on the charge parameter, but which is independent on any mode decomposition. This is to be compared with previous choices appeared in Schwarzschild, for example in \cite{BS03}\cite{DR09} (whose proof makes use of the decomposition in spherical harmonics\footnote{In \cite{DR07}, new techniques avoided the need for spherical harmonics decomposition by using higher-order energy estimates and commuting with angular momentum operators.}) and \cite{DHR19}\cite{HKW20}, which is independent of mode decomposition.  
See also \cite{HMV} for choices in various spherically symmetric spacetimes.  For choices in extremal Reissner-Nordstr\"om see \cite{Aretakis11} \cite{Aretakis11a} and for choices in axial symmetry in subextremal Kerr see Stogin \cite{Stogin17}. For the case of extremal Kerr in axial symmetry see \cite{Aretakis12} for a choice through frequency decomposition, and \cite{GiorgiWan22} by the authors for a proof in physical-space, following Stogin's choice of multiplier. 
Purely physical-space analysis have also been obtained in nonlinear contexts in \cite{KS20}\cite{DHRT21}\cite{GKS}. 
Notice that most of these choices make use of a Hardy inequality to obtain positivity of the spacetime energy obtained.

Since we focus on slowly rotating, strongly charged sub-extremal Kerr-Newman spacetimes, it suffices by continuity to obtain the desired Morawetz estimate for the full sub-extremal Reissner-Nordstr\"om case\footnote{For the analysis of Regge-Wheeler system on sub-extremal Reissner-Nordstr\"om, we do not need to assume axial symmetry.}, without making use of any spherical harmonics decomposition. Our strategy for selecting multipliers in sub-extremal Reissner-Nordstr\"om spacetimes is novel compared to the existing literature. We employ both elliptic estimates on spheres and the method of completing the square to avoid the use of spherical harmonics decomposition. Additionally, we do not need to use any Hardy inequality, ensuring that the final coercivity property depends linearly on a carefully chosen scalar function.

We make use of the multiplier
\begin{equation}
    Y=\frac{\De}{(r^2+a^2)^2}\Big[\big(1-\frac{Q^2}{M^2}\big) (r^2-r_P^2)+\frac{Q^2}{M^2}\frac{r^3-r_P^3}{r} \Big] \partial_r,
\end{equation}
where $r_P$ is the trapping radius. The choice of the multiplier can be thought of as an interpolation between a multiplier that works in the Schwarzschild case ($Y=\frac{\De}{(r^2+a^2)^2}(r^2-r_P^2)\partial_r$) and one that works in the extremal Reissner-Nordstr\"om case ($Y=\frac{\De}{(r^2+a^2)^2}\frac{r^3-r_P^3}{r}\partial_r$).
This choice allows to obtain a positive definite spacetime norm degenerate at trapping, without the need of a Hardy inequality. 
This, combined with the energy estimates in \cite{Giorgi9}, implies the desired boundedness and pointwise decay for the axially symmetric Teukolsky system in the range $|a|\ll M, |Q|<M$.

\medskip

{\bf Acknowledgments:} The first author acknowledges the support of NSF Grants No. DMS-2306143, NSF CAREER grant DMS-2336118 and of a grant of the Simons Foundation (825870, EG).

\section{Preliminaries}

In this section, we recall the main properties of the Kerr-Newman spacetimes and the Teukolsky and gRW system governing coupled gravitational-electromagnetic perturbations of Kerr-Newman. For more details, see Section 2 of \cite{Giorgi9}.

\subsection{The Kerr-Newman metric}\label{sec:KN-metric}

We study the (stationary, axisymmetric) black hole exterior of the Kerr-Newman black hole, whose metric in Boyer-Lindquist coordinates $(t, r,  \th, \phi)$ takes the form 
\bea\label{metric-KN}
g_{a, Q, M}=-\frac{\Delta}{|q|^2}\left( dt- a \sin^2\th d\phi\right)^2+\frac{|q|^2}{\Delta}dr^2+|q|^2 d\th^2+\frac{\sin^2\th}{|q|^2}\left(a dt-(r^2+a^2) d\phi \right)^2,
\eea
where
\beaa
q=r+ i a \cos\th \label{definition-q}, \quad |q|^2=r^2+a^2\cos^2\th, \quad \Delta = (r-r_{+}) (r-r_{-}), \quad r_{\pm}=M\pm \sqrt{M^2-a^2-Q^2},
\eeaa
for each $a$, $Q$, $M$ satisfying $a^2+Q^2 <M^2 $.
We recall that the ambient manifold with boundary $\MM$ is diffeomorphic to $\mathbb{R}^{+} \times \mathbb{R} \times \mathbb{S}^2$, and the metric \eqref{metric-KN} in Kerr star coordinates (see for example \cite{DHR19a}), with $\frac{dr_*}{dr}=\frac{1}{\Delta}$, extends smoothly to the event horizon $\mathcal{H}^+$ defined as the boundary $\partial \MM=\{ r=r_{+}\}$.

We consider a foliation by axially symmetric, smooth spacelike hypersurfaces $\Sigma_\tau$ which connect the event horizon and future null infinity, defined as the time-translated $\Sigma_\tau=\phi_\tau (\Sigma_0)$ associated with the integral curve of the Killing vector field $T=\partial_t$. The region between $\Sigma_{\tau_1}$ and $\Sigma_{\tau_2}$ is denoted $\MM(\tau_1, \tau_2)$.

In the present paper, we will restrict to axially symmetric solutions, which allow crucial simplifications in the analysis. For example, the trapping region is only limited to the radius $\{r=r_P\}$ of trapped null geodesics with zero angular momentum, given by the largest root of the polynomial 
\cite{Giorgi8}
\bea\label{eq:definition-TT}
\TT:= r^3-3Mr^2 + ( a^2+2Q^2)r+Ma^2=0.
\eea 
Also, the ergoregion is practically absent and the Killing vectorfield $\partial_t$ is associated to a coercive energy current.

\subsection{The Teukolsky and the generalized Regge-Wheeler system}

In \cite{Giorgi7}, the first named author defined a set of gauge-invariant quantities associated to a linear electromagnetic-gravitational perturbation of the Kerr-Newman spacetime\footnote{An equivalent set of gauge invariant quantities, denoted by $\Ab, \mathfrak{\underline{F}}, \underline{\mathfrak{B}}, \underline{\mathfrak{X}}$, exist for tensors with negative signature.}:
\beaa
A, \qquad \Ffr, \qquad \mathfrak{B}, \qquad \mathfrak{X},
\eeaa
satisfying a coupled system of Teukolsky equations in Kerr-Newman spacetime,
for their exact definition\footnote{Observe that in the case of vacuum $A$ correspond to the complex tensorial version of the Teukolsky variable $\alpha^{[+2]}$ of spin $+2$. The other quantities do not have an analogue in the vacuum case.} see \cite{Giorgi7}. The one-form $\mathfrak{B}$ and the traceless two tensor $\mathfrak{F}$ admit two derived quantities $\pf$ and $\qf$, at the level of one derivative of $\mathfrak{B}$ and $\mathfrak{F}$ respectively through the so-called Chandrasekhar transformation, which satisfy a 
 generalized Regge-Wheeler system.

\begin{theorem}[Theorem 2.2 in \cite{Giorgi9}, see also Theorem 7.3 in \cite{Giorgi7}]\label{main-theorem-RW} Consider a linear electromagnetic-gravitational perturbation of Kerr-Newman spacetime with mass $M$, charge $Q$ and rotation $a$, and its associated complex gauge invariant quantities $A, \Ffr, \Bfr, \Xfr$.
Then the complex gauge-invariant quantities $\pf $ and $\qf$
 satisfy
 the following coupled system of wave equations:
\bea
 \squared_1\pf-i  \frac{2a\cos\th}{|q|^2}\nab_T \pf  -V_{1,f}  \pf &=&4Q^2 \frac{\ov{q}^3 }{|q|^5} \left(  \ov{\DD} \c  \qf \right) + L_1 \label{final-eq-1}\\
\squared_2\qf -i  \frac{4a\cos\th}{|q|^2}\nab_T \qf -V_{2,f}  \qf &=&-   \frac 1 2\frac{q^3}{|q|^5} \left(  \DD \hot  \pf  -\frac 3 2 \left( H - \Hb\right)  \hot \pf \right) + L_2\label{final-eq-2}
 \eea
where 
\begin{itemize}
\item $\squared_1$ and $\squared_2$ denote the wave operator for horizontal 1- and 2-tensors respectively,
\item $\nab_T$ denote the horizontal covariant derivative with respect to $T=\partial_t$,
\item $V_{1,f}$ and $V_{2,f}$ are real positive potentials explicitly given by
 \beaa
 V_{1,f}=\frac{1}{r^2}\big(1-\frac{2M}{r}+\frac{6Q^2}{r^2} \big)+O(a^2r^{-4}), \qquad V_{2,f}=\frac{4}{r^2}\big(1-\frac{2M}{r}+\frac{3Q^2}{2r^2} \big)+O(a^2r^{-4}),
 \eeaa
 \item $\DD\hot$ and $\DDov$ are horizontal operators given by
 \beaa
\DD\hot(f+i\dual f) &:=& (\nabla+i\dual\nabla)\hot(f+i\dual f)=2\big(\nab \hot f + i \dual (\nab \hot f)\big)\\
\DDov (u+i\dual u) &:=& (\nabla- i\dual\nabla) (u+i\dual u)= 2\big( \div u + i \dual(\div u) \big),
\eeaa
\item $H$ and $\Hb$ are complexified Ricci coefficients,
 \item $L_1$ and $L_2$ denote lower order terms depending on up to one derivatives of $\Bfr$, $\Ffr$, $A$, $\Xfr$.
 \end{itemize}
\end{theorem}

In \cite{Giorgi9}, energy-Morawetz estimates for the gRW system \eqref{final-eq-1}-\eqref{final-eq-2} have been obtained for $|a|, |Q|\ll M$. Here, we restrict to axially symmetric solutions with $|a| \ll M$ and $|Q|<M$. Even though the lower order terms $L_1$ and $L_2$ are simplified in axial symmetry, we do not make advantage of such simplification since those were estimated in \cite{Giorgi9} for $|a| \ll M$, with no restriction on the charge, see Section \ref{sec:lot}.

\subsection{The combined energy-momentum tensor}

Following \cite{Giorgi9}, to study the gRW system \eqref{final-eq-1}-\eqref{final-eq-2}, we consider the following model problem for a one-form and a symmetric traceless two tensors $\psi_1$ and $\psi_2$ respectively:
\bea
 \squared_1\psi_1  -V_1  \psi_1 &=&i  \frac{2a\cos\th}{|q|^2}\nab_T \psi_1+4Q^2 C_1[\psi_2]+ N_1 \label{final-eq-1-model}\\
\squared_2\psi_2 -V_2  \psi_2 &=&i  \frac{4a\cos\th}{|q|^2}\nab_T \psi_2-   \frac {1}{ 2} C_2[\psi_1]+ N_2\label{final-eq-2-model}
 \eea
 where
 \bea\label{eq:potentials-model}
 V_1=\frac{1}{|q|^2}\big(1-\frac{2M}{r}+\frac{6Q^2}{r^2} \big), \qquad V_2=\frac{4}{|q|^2}\big(1-\frac{2M}{r}+\frac{3Q^2}{2r^2} \big),
 \eea
 and the coupling operators are given by 
 \bea\label{eq:C_1-C_2}
 C_1[\psi_2]=\frac{\ov{q}^3 }{|q|^5} \left(  \ov{\DD} \c  \psi_2  \right) , \qquad C_2[\psi_1]=\frac{q^3}{|q|^5} \left(  \DD \hot  \psi_1 -\frac 3 2 \left( H - \Hb\right)  \hot \psi_1 \right).
 \eea
Here $N_1$ and $N_2$ are for now some unspecified right hand sides of the equations.

The energy-momentum tensor for a complex horizontal tensor $\psi$ is defined as
\bea\label{def:energy-momentum-tensor}
\QQ[\psi]_{\mu\nu}:= \Re\big(\Db_\mu  \psi \c \Db _\nu \ov{\psi}\big)
          -\frac 12 \g_{\mu\nu} \big(\Db_\la \psi\c\Db^\la \ov{\psi} + V\psi \c \ov{\psi} \big),
\eea
where $\Re$ denotes the real part, $\Db$ is the projection to the horizontal structure of the covariant derivative. We also denote the Lagrangian $\LL[\psi]:=\Db_\la \psi\c\Db^\la \ov{\psi} + V\psi \c \ov{\psi}$.

 Let $\psi_1$ and $\psi_2$ be horizontal tensors satisfying the model equations \eqref{final-eq-1-model} and \eqref{final-eq-2-model}. We define the following \textit{combined energy-momentum tensor} for the system:
\bea
\QQ[\psi_1, \psi_2]_{\mu\nu}&:=& \QQ[\psi_1]_{\mu\nu}+8Q^2 \QQ[\psi_2]_{\mu\nu} \label{eq:def-combined-em}
\eea

Let $X$ be a vectorfield and $w$ a scalar\footnote{Observe that, in contrast with \cite{Giorgi9}, here we do not need to use a one form $J$ in the definition of the current.}, we define the following \textit{combined current} for the system:
\bea\label{eq:def-combined-current}
 \PP_\mu^{(X, w)}[\psi_1, \psi_2]&:=& \PP_\mu^{(X, w)}[\psi_1]+8 Q^2 \PP_\mu^{(X, w)}[ \psi_2],
\eea
where
 \bea\label{eq:definition-current}
 \PP_\mu^{(X, w)}[\psi]&:=&\QQ[\psi]_{\mu\nu} X^\nu +\frac 1 2  w \Re\big(\psi \c \Db_\mu \overline{\psi} \big)-\frac 1 4 \pr_\mu w |\psi|^2.
  \eea

  We recall here the structure of the divergence of the combined current from \cite{Giorgi9}. 

\begin{proposition}[Proposition 4.1 in \cite{Giorgi9}] \label{prop:general-computation-divergence-P} Let $\psi_1$ and $\psi_2$ satisfying the model system \eqref{final-eq-1-model}-\eqref{final-eq-2-model}. Then, the combined current defined in \eqref{eq:def-combined-current} satisfies the following divergence identity:
\bea\label{eq:divv-PP}
\begin{split}
\D^\mu \PP_\mu^{(X, w)}[\psi_1, \psi_2]&= \EE^{(X, w)}[\psi_1, \psi_2]+\mathscr{N}_{first}^{(X, w)}[\psi_1,\psi_2]+\mathscr{N}_{coupl}^{(X, w)}[\psi_1,\psi_2]\\
&+\mathscr{N}_{lot}^{(X, w)}[\psi_1,\psi_2]+\mathscr{R}^{(X)}[\psi_1, \psi_2],
\end{split}
\eea
where 
\begin{itemize}
\item the bulk term $\EE^{(X, w)}[\psi_1, \psi_2]$ is given by 
\beaa
\EE^{(X, w)}[\psi_1, \psi_2]&:=& \EE^{(X, w)}[\psi_1] +8Q^2 \EE^{(X, w)}[\psi_2]
\eeaa
where
 \bea\label{eq:EE-X-w-J}
 \EE^{(X, w)}[\psi]  &:=& \frac 1 2 \QQ[\psi]  \c\piX - \frac 1 2 X( V ) |\psi|^2+\frac 12  w \LL[\psi] -\frac 1 4 \square_\g  w |\psi|^2,
 \eea

 \item the term $\mathscr{N}_{first}$ involving the first order term on the RHS of the equations is given by
\bea\label{eq:definition-N-first}
\mathscr{N}_{first}^{(X, w)}[\psi_1,\psi_2]:= - \frac{2a\cos\th}{|q|^2} \Im\Big[ \big(\nabla_X\ov{\psi_1} +\frac 1 2   w \ov{\psi_1}\big)\c  \nab_T \psi_1+ 16Q^2\big(\nabla_X\ov{\psi_2} +\frac 1 2   w \ov{\psi_2}\big)\c  \nab_T \psi_2\Big],
\eea

\item the term $\mathscr{N}_{coupl}$ involving the coupling terms on the RHS of the equations is given by
 \bea\label{eq:N-coupl-1}
 \begin{split}
\mathscr{N}_{coupl}^{(X, w)}[\psi_1,\psi_2]&:=4Q^2 \Re\Big[ \big( \frac{ q^3}{|q|^5} w -X(\frac{ q^3}{|q|^5}) \big)\psi_1 \c(\DD \c\ov{\psi_2} ) +\frac{ q^3}{|q|^5} \psi_1 \c ([\DD \c,\nabla_X]\ov{\psi_2}  \big) \Big] \\
  &-\D_\a \Re \Big(\frac{ 4Q^2q^3}{|q|^5}\psi_1 \c \big(\nabla_X\ov{\psi_2} +\frac 1 2   w \ov{\psi_2} \big) \Big)^\a+\nab_X \Re\Big(\frac{ 4Q^2q^3}{|q|^5}\psi_1 \c  (\DD \c\ov{\psi_2}) \Big),
  \end{split}
\eea
or also, equivalently,
 \bea\label{eq:N-coupl-2}
 \begin{split}
 \mathscr{N}_{coupl}^{(X, w)}[\psi_1,\psi_2]&=4Q^2 \Re\Bigg[\big( X( \frac{q^3 }{|q|^5})-\frac{q^3 }{|q|^5}   w \big) (\DD\hot \psi_1 \big) \c   \ov{\psi_2}-   \frac{q^3 }{|q|^5} ([\DD\hot,  \nabla_X]\psi_1)  \c   \ov{\psi_2} \\
&-\Big(\Big(\big( X( \frac{q^3 }{|q|^5})-\frac{q^3 }{|q|^5}   w \big)\frac 3 2(  H-\Hb)+\frac{ q^3}{|q|^5}  \frac 3 2  \nab_X(  H-\Hb)\Big)\hot  \psi_1\Big)\c   \ov{\psi_2} \Bigg]\\
&+\D_\a \Re\Big( \frac{4Q^2q^3 }{|q|^5}  \big(\nabla_X\psi_1 +\frac 1 2   w \psi_1\big) \c \ov{\psi_2}\Big)^\a\\
&-\nab_X\Re\Big(  \frac{4Q^2q^3 }{|q|^5} (\DD\hot \psi_1 \big) \c   \ov{\psi_2}-\frac{ 4Q^2q^3}{|q|^5}  \frac 3 2  ((  H-\Hb)\hot   \psi_1 )\c   \ov{\psi_2}\Big).
\end{split}
\eea

\item the term $\mathscr{N}_{lot}$ involving the lower order terms on the RHS of the equations is given by
\bea\label{eq:definition-N-lot}
\mathscr{N}_{lot}^{(X, w)}[\psi_1,\psi_2]&=&  \Re\Big[ \big(\nabla_X\ov{\psi_1} +\frac 1 2   w \ov{\psi_1}\big)\c N_1+8Q^2  \big(\nabla_X\ov{\psi_2} +\frac 1 2   w \ov{\psi_2}\big)\c N_2\Big],
\eea
 
 \item the curvature term $\mathscr{R}$ is given by
 \bea\label{eq:definition-R-X}
 \mathscr{R}^{(X)}[\psi_1, \psi_2]&:=&  \Re\Big[ X^\mu \Db^\nu  \psi_1^a\Rdot_{ ab   \nu\mu}\ov{\psi_1}^b+ 8Q^2 X^\mu \Db^\nu  \psi_2^a\Rdot_{ ab   \nu\mu}\ov{\psi_2}^b\Big].
 \eea
\end{itemize}

The above Proposition uses crucially the following.

 \begin{lemma}[Lemma 2.11 in \cite{Giorgi7}]\label{lemma:adjoint-operators}
 For $F=f+i\dual f \in \sk_1(\CCC)$ and $U=u+i\dual u \in \sk_2(\CCC)$, we have
   \bea
 ( \DD \hot   F) \c   \ov{U}  &=&  -F \c (\DD \c \ov{U}) -( (H+\Hb ) \hot F )\c \ov{U} +\D_\a (F \c \ov{U})^\a.
 \eea
 \end{lemma}

\end{proposition}

\subsection{Elliptic identities and estimates}

Recall the following elliptic identities  for a 1-tensor $\xi$ and a 2-tensor $\chi$ (see Lemma 4.8.1 in \cite{GKS} and Section 4.4 in \cite{Giorgi9}):
\beaa
 |\nab  \xi   |^2-\Kh |\xi|^2&=2|\DDs_2   \xi   |^2-  \frac{2a\cos\th}{|q|^2}\dual  \nab_T  \xi  \c \xi+\D_\a \big( \nab^\a \xi \c \xi+2 (\DDs_2 \xi)^{\a\b} \xi_{\b} \big)\\
|\nab    \chi |^2+ 2 \Kh |\chi|^2 &=2 |\DDd_2   \chi  |^2+  \frac{2a\cos\th}{|q|^2}\dual  \nab_T  \chi  \c \chi+\D_\a \big( \nab^\a \chi \c \chi-2 (\div \chi)_\b \chi^{\a\b} \big),
\eeaa
and the respective complex analogue:
\bea
 |\nab \psi_1  |^2-\Kh |\psi_1 |^2&=&\frac 1 2 |\DD \hot  \psi_1  |^2+  \frac{2a\cos\th}{|q|^2} \Im(  \nab_T  \psi_1  \c\ov{\psi_1})+\D_\a\Re \big( \nab^\a \psi_1\c \ov{\psi_1}+ (\DD \hot  \psi_1) \c \ov{\psi_1} \big),\label{eq:elliptic-estimates-psi1}\\
 |\nab \psi_2  |^2+2\Kh |\psi_2 |^2&=&\frac 1 2 |\ov{\DD}\c  \psi_2  |^2-  \frac{2a\cos\th}{|q|^2} \Im(  \nab_T  \psi_2  \c\ov{\psi_2})+\D_\a\Re \big( \nab^\a \psi_2\c \ov{\psi_2}+ (\DD \hot  \psi_2) \c \ov{\psi_2} \big).\label{eq:elliptic-estimates-psi2}
\eea

We also recall the following Poincar\'e inequality, involving the integral on the topological spheres in Kerr-Newman spanned by $(\theta, \phi)$.

\begin{lemma}[Lemma 7.2.3 in \cite{GKS} or Lemma 4.5 in \cite{Giorgi9}]\lab{lemma:poincareinequalityfornabonSasoidfh:chap6}
For $\psi_1$ a one-form, we have 
\beaa
\int_S|\nab\psi_1|^2 &\geq& \frac{1}{r^2}\int_S|\psi_1|^2  -O(a)\int_S\big(|\nab\psi_1|^2+r^{-2}|\nab_T\psi_1|^2+r^{-4}|\psi_1|^2\big).
\eeaa
For $\psi_2$ a symmetric traceless 2-tensor, we have 
\beaa
\int_S|\nab\psi_2|^2 &\geq& \frac{2}{r^2}\int_S|\psi_2|^2  -O(a)\int_S\big(|\nab\psi_2|^2+r^{-2}|\nab_T\psi_2|^2+r^{-4}|\psi_2|^2\big).
\eeaa
\end{lemma}

Combining the above, we also obtain
     \bea\label{eq:elliptic-nablapsi2-divpsi2}
      \int_S     |\nab \psi_2  |^2            &\geq& \frac 1 4 \int_{S}   |\ov{\DD}\c  \psi_2  |^2 -O(a)\int_S\big(|\nab\psi_2|^2+r^{-2}|\nab_T\psi_2|^2+r^{-2}|\psi_2|^2\big),
     \eea
We will use the notation  $\geq_S$ to denote that the inequality holds upon integration on spheres, as above.

\section{Statement of the Theorem and overview of the proof}

We define the following energies on $\Sigma_\tau$:
 \beaa
 E[\pf , \qf](\tau) &:=&\int_{\Si_\tau} \left( |\nab_T\pf|^2+ |\nab_T\qf|^2 +   |\nab_{\pr_{r_*}}\pf|^2 +   |\nab_{\pr_{r_*}}\qf|^2+|\nab\pf|^2 +|\nab\qf|^2 + r^{-2} \big( |\pf|^2  + |\qf|^2\big)\right), 
 \eeaa
 and
 \beaa
E[\Bfr, \Ffr](\tau)&:=&\int_{\Si_\tau}  r^{8}\big(|\nab_4\Bfr|^2+|\nab_{3}\Bfr|^2+r^{-1} |\nab \Bfr|^2+r^{-2}|\Bfr|^2\big) \\
&&+\int_{\Si_\tau}r^{4}\big(|\nab_4\Ffr|^2+|\nab_{3}\Ffr|^2+r^{-1} |\nab \Ffr|^2+r^{-2}|\Ffr|^2\big), \\
E[A, \Xfr](\tau)&:=& \int_{\Si_\tau} r^{3}\big(   |\nab_3A|^2+  |\nab A|^2+ r^{-1}  |A|^2+ |\nab_3\Xfr|^2+ |\nab \Xfr|^2+r^{-1} |\Xfr|^2\big).
\eeaa
Here $\nab_3$ and $\nab_4$ denote the projected covariant derivative with respect to the ingoing and outgoing null directions respectively.

We denote the combined energy 
\beaa
E[\Bfr, \Ffr, A, \Xfr](\tau)&:=&  E[\Bfr, \Ffr](\tau)+E[A, \Xfr](\tau), \\
E[\pf, \qf, \Bfr, \Ffr, A, \Xfr](\tau)&:=&  E[\pf , \qf](\tau)+E[\Bfr, \Ffr](\tau)+E[A, \Xfr](\tau).
\eeaa

We also define the following spacetime energies on $\MM(\tau_1, \tau_2)$:
 \beaa
      \Mor^{ax}[\pf, \qf](\tau_1, \tau_2)&:=&\int_{\MM(\tau_1, \tau_2)}   r^{-3} \big(  |\nab_{\pr_{r_*}}\pf|^2+ |\nab_{\pr_{r_*}}\qf|^2+|\pf|^2+|\qf|^2\big) \\
     && +\int_{\MM(\tau_1, \tau_2)}\left(1-\frac{r_P}{r}\right)^2  \Big(  r^{-2} |\nab_{T} \pf|^2+  r^{-2} |\nab_{T} \qf|^2 +r^{-1} |\nab \pf|^2+ r^{-1} |\nab \qf|^2  \Big)
 \eeaa
and 
 \beaa
B[\Bfr, \Ffr](\tau_1, \tau_2)&:=& \int_{\MM(\tau_1, \tau_2)}     r^{7} \big(|\nab_4\Bfr|^2+|\nab_3\Bfr|^2+|\nab \Bfr|^2+r^{-2}|\Bfr|^2\big)\\
&&+\int_{\MM(\tau_1, \tau_2)} r^{3} \big(|\nab_4\Ffr|^2+|\nab_3\Ffr|^2+|\nab \Ffr|^2+r^{-2}|\Ffr|^2\big)\\
B[A, \Xfr](\tau_1, \tau_2)&:=& \int_{\MM(\tau_1, \tau_2)} r^{3}\big(  |\nab_3A|^2+  |\nab A|^2+ r^{-2} |A|^2+ |\nab_3\Xfr|^2+ |\nab \Xfr|^2+r^{-2} |\Xfr|^2\big).
\eeaa
We denote the combined spacetime energy 
\beaa
B[\Bfr, \Ffr, A, \Xfr](\tau_1, \tau_2)&:=&  B[\Bfr, \Ffr](\tau_1, \tau_2)+B[A, \Xfr](\tau_1, \tau_2), \\
B[\pf, \qf, \Bfr, \Ffr, A, \Xfr](\tau_1, \tau_2)&:=&  \Mor^{ax}[\pf , \qf](\tau_1, \tau_2)+B[\Bfr, \Ffr](\tau_1, \tau_2)+B[A, \Xfr](\tau_1, \tau_2).
\eeaa

We can finally state the precise statement of our theorem.

\begin{theorem}
\lab{Thm:Nondegenerate-Morawetz}\lab{theorem:unconditional-result-final}

Let $\Bfr, \Ffr, A, \Xfr$ be \textbf{axisymmetric} solutions of the Teukolsky system and $\pf, \qf$ their Chandrasekhar transformed solutions the gRW system \eqref{final-eq-1}-\eqref{final-eq-2}. Then, for $|a| \ll M$, $0<|Q| <M$  so that $a^2+Q^2<M^2$ the following energy boundedness and integrated local energy decay estimates hold true:
 \bea\label{eq:final-estimate-theorem}
E[\pf, \qf, \Bfr, \Ffr, A, \Xfr](\tau)+B[\pf, \qf, \Bfr, \Ffr, A, \Xfr] (0, \tau) \les  E[\pf, \qf, \Bfr, \Ffr, A, \Xfr](0).
       \eea
\end{theorem}

The estimate \eqref{eq:final-estimate-theorem} was proved in Theorem 3.2 of \cite{Giorgi9} for any solution of the Teukolsky system (outside axial symmetry) in Kerr-Newman spacetimes with $|a|, |Q|\ll M$. As in \cite{Giorgi9}, \eqref{eq:final-estimate-theorem} can be easily upgraded with $r^p$-weighted and higher derivative energies.

\subsection{Overview of the proof}

We give here an overview of the proof.

As in \cite{Giorgi9}, we first consider solutions to the model system \eqref{final-eq-1-model}-\eqref{final-eq-2-model}, which admit the combined energy momentum tensor defined by \eqref{eq:def-combined-em}.

We first obtain energy estimates for the system. By applying the divergence theorem to the combined current $\PP_\mu^{(T, 0)}[\psi_1, \psi_2]$ we obtain a conditional energy of the system, still depending on the right hand side of the equations. Nevertheless,
the assumption of axial symmetry allows to make direct use of such conditional energy estimates as they can be bounded by the degenerate spacetime energy $\Mor^{ax}[\psi_1, \psi_2](0, \tau)$, see Proposition \ref{prop:energy-estimates-conditional}. This is recalled in Section \ref{sec:energy}.

The main part of the proof is then the derivation of the Morawetz estimates, i.e. the choice of a multiplier for the combined current $\PP_\mu^{(\FF(r) \partial_r, w(r))}[\psi_1, \psi_2]$ which gives a coercive spacetime energy norm. This is achieved with the choice $\FF(r)=z(r)u(r)$ with $z=\frac{\De}{(r^2+a^2)^2}$ and 
\beaa
u=\big(1-\frac{Q^2}{M^2}\big) (r^2-r_P^2)+\frac{Q^2}{M^2}\frac{r^3-r_P^3}{r}, \qquad w=\frac{\De}{(r^2+a^2)^2}\pr_r u.
\eeaa
Because of the right hand side of the equations in the system, the spacetime energy contains mixed terms, such as $\psi_1 \c (\DD \c \ov{\psi_2})$. We use elliptic estimates to obtain a coercive lower bound for the quadratic form, modulo lower order terms in $a$. This is done in Section \ref{sec:morawetz}.

Finally, we pass from the estimates for the model system to the ones for the gRW system by using previous controls on the lower order terms $L_1$ and $L_2$ obtained in Propositions 3.5 and 3.6 in \cite{Giorgi9}. This is recalled in Section \ref{sec:lot}.

\section{Proof of the Theorem}

Here we prove Theorem \ref{Thm:Nondegenerate-Morawetz}.

\subsection{Energy estimates}\label{sec:energy}

Energy estimates for the model system which hold true in the range $|a| \ll M$, $|Q|<M$ were obtained in \cite{Giorgi9}. We recall here the statement.

\begin{proposition}[Proposition 3.5 in \cite{Giorgi9}]\label{prop:energy-estimates-conditional} Let $\psi_1, \psi_2$ be solutions of the model system \eqref{final-eq-1-model}-\eqref{final-eq-2-model}. For $0<|Q|<M$, $|a|\ll M$, the following conditional energy boundedness estimate holds true:
 \bea
 \label{eq:energy-estimates-conditional}
 \begin{split}
E[\psi_1, \psi_2](\tau)   &\les  E[\psi_1, \psi_2](0)+|a|  \Mor^{ax}[\psi_1, \psi_2](0, \tau) \\
 &  +\left|\int_{\MM(0, \tau)}\Re\Big( \nabla_{T}\ov{\psi_1} \c N_1+8Q^2  \nabla_{T}\ov{\psi_2}  \c N_2\Big)\right|+\int_{\MM(0, \tau)}\Big( |N_1|^2+Q^2 |N_2|^2\Big).
 \end{split}
 \eea
\end{proposition}

\begin{proof}
Observe that Proposition 3.5 in \cite{Giorgi9} is stated with the Morawetz bulk $\Mor[\psi_1, \psi_2]$, given by 
   \beaa
 \Mor[\psi_1, \psi_2](\tau_1, \tau_2)&=&\int_{\MM(\tau_1, \tau_2) } 
  \left(    r^{-2} \big( | \nab_{\pr_{r_*}}  \psi_1|^2+|\nab_{\pr_{r_*}} \psi_2|^2\big) +r^{-3}\big(|\psi_1|^2+  |\psi_2|^2\big) \right)\\
      &&+ \int_{\Mntrap(\tau_1, \tau_2)} \left(  r^{-2}|\nab_3\psi_1|^2+ r^{-2} |\nab_3 \psi_2|^2 + r^{-1}  |\nab  \psi_1|^2+ r^{-1} |\nab \psi_2|^2\right),
\eeaa
in lieu of the axially symmetric Morawetz bulk $\Mor^{ax}[\psi_1, \psi_2](0, \tau)$. Here, $\Mntrap$ denotes the region of the spacetime outside the open trapping region.
Since the axially symmetric Morawetz bulk controls the Morawetz bulk, i.e. $\Mor^{ax}[\psi_1, \psi_2](\tau_1, \tau_2)\gtrsim \Mor[\psi_1, \psi_2](\tau_1, \tau_2)$, the above is immediately implied by Proposition 3.5 in \cite{Giorgi9}.

Observe that even though some parts of the proof are simplified in axial symmetry by applying Proposition \ref{prop:general-computation-divergence-P} with $X=T$, $w=0$, such as $ \EE^{(T, 0)}[\psi_1, \psi_2]=0$ and $\mathscr{N}_{first}^{(T, 0)}[\psi_1,\psi_2]=0$, in the remaining terms we still have the presence of terms such as $O(ar^{-5})  |\psi_1||\psi_2|$, which need to be bounded on the right hand side by $\Mor^{ax}[\psi_1, \psi_2]$.
\end{proof}

Observe that because of the coupling terms in the model system, the energy estimates are still conditional with respect to the axially symmetric Morawetz bulk, in contrast with the case of Reissner-Nordstr\"om \cite{Giorgi7}. This is due for instance to the coupling term $\mathscr{N}_{coupl}^{(T, 0)}[\psi_1,\psi_2]$ and the curvature term $\mathscr{R}^{(T)}[\psi_1, \psi_2]$ and the fact that the commutators between $T$ and angular derivative pick up terms which are $O(|a|)$.

\subsection{Morawetz estimates}\label{sec:morawetz}

Consider the vectorfield $Y=\FF(r) \partial_r$ for a well-chosen function $\FF$, together with a scalar function $w_Y$. To obtain Morawetz estimates for the model system, we apply Proposition \ref{prop:general-computation-divergence-P} with the above and obtain
 \bea\label{eq:divergence-theorem-identity-Y}
\begin{split}
\D^\mu \PP_\mu^{( Y, w_Y)}[\psi_1, \psi_2]&= \EE^{(Y, w_Y)}[\psi_1, \psi_2] +\mathscr{N}_{first}^{( Y, w_Y)}[\psi_1,\psi_2]+\mathscr{N}_{coupl}^{( Y, w_Y)}[\psi_1,\psi_2]\\
&+\mathscr{N}_{lot}^{(Y, w_Y)}[\psi_1,\psi_2]+\mathscr{R}^{(Y)}[\psi_1, \psi_2].
\end{split}
\eea
In Section 6.1 of \cite{Giorgi9}, we computed the above terms for 
\bea\label{eq:FF-wred-w}
\FF=z u , \qquad w_Y = z \pr_r u, 
\eea
where $z$ and $u$ are any two functions, giving 
\bea\label{eq:EE-Y-wY-J}
\begin{split}
|q|^2\EE^{(Y, w_Y)}[\psi_1, \psi_2]&=\AA \big(  |\nab_r\psi_1|^2 + 8Q^2|\nab_r\psi_2|^2 \big) + \UU^{\a\b} \Re\big( \Db_\a \psi_1 \c\Db_\b \ov{\psi_1} +8Q^2 \Db_\a \psi_2 \c\Db_\b \ov{\psi_2} \big)\\
&+\big( \VV_1 |\psi_1|^2+ 8Q^2 \VV_2 |\psi_2|^2 \big),
\end{split}
\eea
with $\AA$, $\UU^{\a\b}$, $\VV_1$, $\VV_2$ given by  
   \bea\label{eq:expressions-AA-UU-VV-general-hf}\lab{eq:coeeficientsUUAAVV-u}
\begin{split}
 \AA&=z^{1/2}\Delta^{3/2} \partial_r\left( \frac{ z^{1/2}  u}{\Delta^{1/2}}  \right),   \\
  \UU^{\a\b}&= -  \frac{ 1}{2}  u \pr_r\left( \frac z \De\RR^{\a\b}\right),\\
\VV_i&= -  \frac 1 4  \pr_r\Big(\De \pr_r \big(
 z \pr_r u  \big)  \Big)-\frac 1 2  u  \pr_r \left(z |q|^2 V_i\right)=: V_0 + V_{pot, i} ,
 \end{split}
\eea
where $V_0=- \frac 1 4  \pr_r\Big(\De \pr_r \big(
 z \pr_r u  \big)  \Big)$ and $V_{pot, i}=-\frac 1 2  u \pr_r \left(z |q|^2 V_i\right)$ for $i=1,2$.

 Choosing, as\footnote{The choice of $u$ will instead differ from \cite{Giorgi9}, allowing to obtain positivity in the range $0<|Q|<M$.} in \cite{Giorgi9}, 
\bea\label{eq:choice-z}
z=\frac{\De}{(r^2+a^2)^2},
\eea
the above coefficients are given by
   \bea\label{eq:expressions-AA-UU-VV-general-hf-withz}
\begin{split}
 \AA&=\frac{\Delta^{2}}{r^2+a^2} \partial_r\left( \frac{   u}{r^2+a^2}  \right),   \\
  \UU^{\a\b}&= -  \frac{ 1}{2}  u \pr_r\left(  \frac{1}{(r^2+a^2)^2}\RR^{\a\b}\right),\\
  &= -  \frac{ 1}{2}  u \pr_r\left(-  \frac{2a}{(r^2+a^2)}\partial_t^{(\a} \partial_\phi^{\b)}- \frac{a^2}{(r^2+a^2)^2}  \partial_\phi^\a \partial_\phi^\b+ \frac{\Delta}{(r^2+a^2)^2}  O^{\a\b}\right),\\
    &= -  \frac{ 1}{2}  u \left(   \frac{4ar}{(r^2+a^2)^2}  \pr_t^{(\a} \pr_\phi^{\b)} +\frac{4a^2r}{(r^2+a^2)^3} \pr_\phi^\a \pr_\phi^\b-\frac{2\TT}{(r^2+a^2)^3}  O^{\a\b}\right),\\
\VV_i&= -  \frac 1 4  \pr_r\Big(\De \pr_r 
 w_Y   \Big)-\frac 1 2  u  \pr_r \left(\frac{\De}{(r^2+a^2)^2} |q|^2 V_i\right),
 \end{split}
\eea
where we used that 
\beaa
\RR^{\a\b}&=&  -(r^2+a^2)^2 \partial_t^\a \partial_t^\b-2a(r^2+a^2)\partial_t^{(\a} \partial_\phi^{\b)}-a^2  \partial_\phi^\a \partial_\phi^\b+ \Delta O^{\a\b}, \label{definition-RR-tensor}\\
 O^{\a\b}&=& \partial_\th^\a  \partial_\th^\b  +\frac{1}{\sin^2\th} \partial_{\vphi}^\a \partial_{\vphi}^\b+2a\partial_t^{(\a} \partial_\vphi^{\b)}+a^2 \sin^2\th \partial_t^\a \partial_t^\b, \nn
\eeaa 
and $\pr_rz=-\frac{2\TT}{(r^2+a^2)^3}$. We denote $ O^{\a\b} \nab_\a\psi\c\nab_\b\psi =: |q|^2 |\nab\psi|^2$.

Observe that in axial symmetry, the term $\UU^{\a\b}$ reduces to 
   \bea\label{eq:expression-U-ab}
\begin{split}
  \UU^{\a\b}
    &=   \frac{u \TT}{(r^2+a^2)^3}  O^{\a\b}
 \end{split}
\eea

Also, from \cite{Giorgi9}, we have, with the choice \eqref{eq:choice-z},
\bea\label{eq:mathscr-NN-coupl-X}
\begin{split}
\mathscr{N}_{coupl}^{(Y, w_Y)}[\psi_1,\psi_2]  &=4Q^2 \Re\Big[ u \big( \frac{zq^3(3q-2r)}{|q|^7}-\frac{(\pr_rz) q^3}{|q|^5} \big)  \psi_1 \c \left(  \DD \c  \ov{\psi_2}  \right)\Big] +O(a^2r^{-6})\FF \Re( \psi_1 \c \ov{\psi_2})\\
  &-\D_\mu \Re\Big[\frac{4Q^2 q^3}{|q|^5} \big( \psi_1 \c \big(\nabla_Y\ov{\psi_2} +\frac 1 2   w_Y \ov{\psi_2} \big)\big)^\mu -\frac{4Q^2z q^3}{|q|^5}u\psi_1 \c  (\DD \c\ov{\psi_2}) \pr_r^\mu \Big]\\
  &=\frac{8Q^2}{|q|^5(r^2+a^2)^2} \Re\Big[ uq^3 \big( \frac{\TT}{r^2+a^2}+\frac{\De(r+3ia\cos\th)}{2|q|^2} \big)  \psi_1 \c \left(  \DD \c  \ov{\psi_2}  \right)\Big] +O(a^2r^{-6})\FF \Re( \psi_1 \c \ov{\psi_2})\\
  &-\D_\mu \Re\Big[\frac{4Q^2 q^3}{|q|^5} \big( \psi_1 \c \big(\nabla_Y\ov{\psi_2} +\frac 1 2   w_Y \ov{\psi_2} \big)\big)^\mu -\frac{4Q^2z q^3}{|q|^5}u\psi_1 \c  (\DD \c\ov{\psi_2}) \pr_r^\mu \Big].
  \end{split}
\eea
and
    \bea\label{eq:mathscr-N-first0general}
    \begin{split}
\mathscr{N}_{first}^{(Y, w_Y)}[\psi_1,\psi_2]&=  \frac{a\cos\th}{|q|^2} \Big[  (\pr_r z) u \Im\Big(\ov{\psi_1}\c \nab_T \psi_1+16 Q^2\ov{\psi_2}\c \nab_T \psi_2 \Big) +2zu  \rhod\frac{|q|^2}{\De}(|\psi_1|^2+16Q^2 |\psi_2|^2\big)\Big]\\
   &-\D_\mu \Im \Big[  \frac{a\cos\th}{|q|^2} (\partial_r)^\mu   zu (\ov{\psi_1}\c\nab_T \psi_1+16Q^2\ov{\psi_2}\c\nab_T \psi_2 )\Big]  \\
   &+\partial_t \Im \Big(  \frac{a\cos\th}{|q|^2} zu (\ov{\psi_1}\c  \nab_r \psi_1+16Q^2\ov{\psi_2}\c  \nab_r \psi_2)\Big), 
   \end{split}
\eea
 \bea\label{eq:mathsc-R-Y}
 \begin{split}
 \mathscr{R}^{(Y)}[\psi_1, \psi_2]&= \left(\big(\rhod +\etab\wedge\eta\big)\frac{r^2+a^2}{\De}+\frac{2a^3r\cos\th(\sin\th)^2}{|q|^6}\right)\FF \Re\big( i \nab_{\That}\psi_1\c\ov{\psi_1}+i8Q^2 \nab_{\That}\psi_2\c\ov{\psi_2}\big).
\end{split}
 \eea

We rewrite \eqref{eq:divergence-theorem-identity-Y} as 
 \bea
\begin{split}
\D^\mu \PP_\mu^{( Y, w_Y)}[\psi_1, \psi_2]&=\mbox{Main}+\mbox{O(a) terms},
\end{split}
\eea
where
 \bea
\begin{split}
\mbox{Main}&:= \EE^{(Y, w_Y)}[\psi_1, \psi_2]+\mathscr{N}_{coupl}^{( Y, w_Y)}[\psi_1,\psi_2]\\
\mbox{O(a) terms}&:=\mathscr{N}_{first}^{( Y, w_Y)}[\psi_1,\psi_2]+\mathscr{N}_{lot}^{(Y, w_Y)}[\psi_1,\psi_2]+\mathscr{R}^{(Y)}[\psi_1, \psi_2].\label{eq:definition-Oaterms}
\end{split}
\eea
We now focus on $\mbox{Main}$. By using the expressions in \eqref{eq:expressions-AA-UU-VV-general-hf-withz}, \eqref{eq:expression-U-ab} and \eqref{eq:mathscr-NN-coupl-X}, we regroup the terms as:
 \beaa
\mbox{Main}&=&\frac{\Delta^{2}}{|q|^2(r^2+a^2)} \partial_r\left( \frac{   u}{r^2+a^2}  \right) \big(  |\nab_r\psi_1|^2 + 8Q^2|\nab_r\psi_2|^2 \big)+\mbox{I} +\mbox{II}\\
  &&-\D_\mu \Re\Big[\frac{4Q^2 q^3}{|q|^5} \big( \psi_1 \c \big(\nabla_Y\ov{\psi_2} +\frac 1 2   w_Y \ov{\psi_2} \big)\big)^\mu -\frac{4Q^2z q^3}{|q|^5}u\psi_1 \c  (\DD \c\ov{\psi_2}) \pr_r^\mu \Big]+O(a^2r^{-6})\FF \Re( \psi_1 \c \ov{\psi_2})
\eeaa
where
\beaa
\mbox{I}&:=&\frac{u \TT}{(r^2+a^2)^3} \big( |\nab \psi_1|^2 +8Q^2 |\nab \psi_2 |^2 \big)+\frac{8Q^2}{|q|^5(r^2+a^2)^2} \Re\Big[ uq^3 \big( \frac{\TT}{r^2+a^2}+\frac{\De(r+3ia\cos\th)}{2|q|^2} \big)  \psi_1 \c \left(  \DD \c  \ov{\psi_2}  \right)\Big] , \\
\mbox{II}&:=&\frac{1}{|q|^2}\big( \VV_1 |\psi_1|^2+ 8Q^2 \VV_2 |\psi_2|^2 \big).
\eeaa

We have the following general bound for terms $\mbox{I}+\mbox{II}$.

\begin{lemma} For any $p_1,p_2\in(0,1]$ and $\delta\in (0,1)$ we have
    \bea \label{eq:expression-I-II}   \mbox{I}+\mbox{II}&\geq_S&\delta \frac{u \TT}{(r^2+a^2)^3} \big( |\nab \psi_1|^2 +8Q^2 |\nab \psi_2 |^2 \big)+A_1 |\psi_1|^2+8Q^2 A_2|\psi_2|^2\\
&&-O(a)\frac{u \TT}{(r^2+a^2)^3}\big(|\nab\psi_1|^2+r^{-2}|\nab_T\psi_1|^2+r^{-2}|\psi_1|^2+|\nab\psi_2|^2+r^{-2}|\nab_T\psi_2|^2+r^{-2}|\psi_2|^2\big)\nonumber\\
&&-O(a) u r^{-6}\Big( \Re[ \psi_1 \c \left(  \DD \c  \ov{\psi_2}  \right)]+\Im[ \psi_1 \c \left(  \DD \c  \ov{\psi_2}  \right)]+r^{-1}|\psi_1|^2+r^{-1}|\psi_2|^2\Big)\nonumber
    \eea
 where 
 \begin{align}
     A_1&:= (1-\delta) \frac{u \TT}{(r^2+a^2)^3}  \frac{1}{r^2}  - \frac{8p_1^2Q^2r^6}{(1-\delta)p_2|q|^{10}(r^2+a^2)}\frac{u}{ \TT}\Big( \frac{\TT}{r^2+a^2}+\frac{\De r}{2|q|^2} \Big)^2 \nonumber\\
     &+\frac{1}{|q|^2}\Big( -  \frac 1 4  \pr_r\big(\De \pr_r 
 w_Y   \big)-\frac 1 2  u  \pr_r \big(\frac{\De}{(r^2+a^2)^2} |q|^2 V_1\big) \Big) \label{eq:definition-A1}, \\
     A_2&:=(1-\delta) \frac{2(1-p_2)u \TT}{(r^2+a^2)^3}  \frac{1}{r^2}-\frac{4(1-p_1)^2Q^2r^6}{(1-\delta)|q|^{10}(r^2+a^2)}\frac{u}{ \TT} \Big( \frac{\TT}{r^2+a^2}+\frac{\De r}{2|q|^2} \Big)^2 \nonumber\\
     &+\frac{1}{|q|^2}\Big( -  \frac 1 4  \pr_r\big(\De \pr_r 
 w_Y   \big)-\frac 1 2  u  \pr_r \big(\frac{\De}{(r^2+a^2)^2} |q|^2 V_2\big) \Big). \label{eq:definition-A2}
 \end{align}
\end{lemma}

\begin{proof}
Using \eqref{eq:elliptic-estimates-psi1}, \eqref{eq:elliptic-nablapsi2-divpsi2} and Lemma \ref{lemma:poincareinequalityfornabonSasoidfh:chap6}, i.e.
\beaa
 |\nab \psi_1  |^2&\geq_S&\frac{1}{r^2} |\psi_1 |^2+\frac 1 2 |\DD \hot  \psi_1  |^2 -O(a)\big(|\nab\psi_1|^2+r^{-2}|\nab_T\psi_1|^2+r^{-2}|\psi_1|^2\big), \\
|\nab\psi_2|^2 &\geq_S& \frac{2}{r^2}|\psi_2|^2  -O(a)\big(|\nab\psi_2|^2+r^{-2}|\nab_T\psi_2|^2+r^{-4}|\psi_2|^2\big), \\
          |\nab \psi_2  |^2            &\geq_S& \frac 1 4    |\ov{\DD}\c  \psi_2  |^2 -O(a)\big(|\nab\psi_2|^2+r^{-2}|\nab_T\psi_2|^2+r^{-2}|\psi_2|^2\big),
      \eeaa
we write for any $p_2 \in [0,1]$
\begin{align*}
|\nab\psi_2|^2 \geq_S 
\frac{2(1-p_2)}{r^2}|\psi_2|^2 +\frac{p_2}{4}|\ov \DD\cdot \psi_2|^2 -O(a)\big(|\nab\psi_2|^2+r^{-2}|\nab_T\psi_2|^2+r^{-4}|\psi_2|^2\big).
\end{align*}

Also, using Lemma \ref{lemma:adjoint-operators} we deduce
    \bea
 ( \DD \hot   \psi_1) \c   \ov{\psi_2}  &=_S&  -\psi_1 \c (\DD \c \ov{\psi_2}) -O(a)r^{-2}\big(|\psi_1|^2+|\psi_2|^2 \big)  +\D_\a (\psi_1 \c \ov{\psi_2})^\a.
 \eea
 and so for any $p_1 \in [0,1]$ we can write
 \begin{align*}
     \Re\Big[ \psi_1 \c \left(  \DD \c  \ov{\psi_2}  \right)\Big]&=p_1\Re\Big[ \psi_1 \c \left(  \DD \c  \ov{\psi_2}  \right)\Big]-(1-p_1)\Re\Big[ ( \DD \hot   \psi_1) \c   \ov{\psi_2}\Big]-O(a)r^{-2}\big(|\psi_1|^2+|\psi_2|^2 \big) . 
 \end{align*}
Therefore using the above we obtain for $q\neq 0$ and $\delta>0$
\beaa
\mbox{I}&=&\frac{u \TT}{(r^2+a^2)^3} \big( |\nab \psi_1|^2 +8Q^2 |\nab \psi_2 |^2 \big)+\frac{8Q^2}{|q|^5(r^2+a^2)^2} \Re\Big[ uq^3 \big( \frac{\TT}{r^2+a^2}+\frac{\De(r+3ia\cos\th)}{2|q|^2} \big)  \psi_1 \c \left(  \DD \c  \ov{\psi_2}  \right)\Big] \\
&\geq_S&\delta \frac{u \TT}{(r^2+a^2)^3} \big( |\nab \psi_1|^2 +8Q^2 |\nab \psi_2 |^2 \big)\\
&&(1-\delta)\frac{u \TT}{(r^2+a^2)^3} \big( \frac{1}{r^2} |\psi_1 |^2+\frac 1 2 |\DD \hot  \psi_1  |^2  \big)+8Q^2(1-\delta) \frac{u \TT}{(r^2+a^2)^3} \big( \frac{2(1-p_2)}{r^2}|\psi_2|^2 +\frac{p_2}{4}|\ov \DD\cdot \psi_2|^2 \big)\\
&&+\frac{8Q^2ur^3}{|q|^5(r^2+a^2)^2}  \big( \frac{\TT}{r^2+a^2}+\frac{\De r}{2|q|^2} \big) \Big(p_1\Re\Big[ \psi_1 \c \left(  \DD \c  \ov{\psi_2}  \right)\Big]-(1-p_1)\Re\Big[ ( \DD \hot   \psi_1) \c   \ov{\psi_2}\Big] \Big)\\
&&+\frac{u \TT}{(r^2+a^2)^3}\Big(-O(a)\big(|\nab\psi_1|^2+r^{-2}|\nab_T\psi_1|^2+r^{-2}|\psi_1|^2+|\nab\psi_2|^2+r^{-2}|\nab_T\psi_2|^2+r^{-2}|\psi_2|^2\big)\Big)\\
&&-u O(aQ^2)r^{-6}\Big( \Re[ \psi_1 \c \left(  \DD \c  \ov{\psi_2}  \right)]+\Im[ \psi_1 \c \left(  \DD \c  \ov{\psi_2}  \right)]+r^{-1}|\psi_1|^2+r^{-1}|\psi_2|^2\Big)\\
&=&\delta \frac{u \TT}{(r^2+a^2)^3} \big( |\nab \psi_1|^2 +8Q^2 |\nab \psi_2 |^2 \big)\\
&&+\frac{1-\delta}{2}\frac{u \TT}{(r^2+a^2)^3}   |\DD \hot  \psi_1  |^2 -\frac{8(1-p_1)Q^2ur^3}{|q|^5(r^2+a^2)^2}  \big( \frac{\TT}{r^2+a^2}+\frac{\De r}{2|q|^2} \big) \Re\Big[ ( \DD \hot   \psi_1) \c   \ov{\psi_2}\Big] \\
&&+ \frac{2(1-\delta)p_2Q^2 u \TT}{(r^2+a^2)^3}|\ov \DD\cdot \psi_2|^2 +\frac{8p_1Q^2ur^3}{|q|^5(r^2+a^2)^2}  \big( \frac{\TT}{r^2+a^2}+\frac{\De r}{2|q|^2} \big) \Re\Big[ \psi_1 \c \left(  \DD \c  \ov{\psi_2}  \right)\Big]\\
&&+(1-\delta)\frac{u \TT}{(r^2+a^2)^3}  \frac{1}{r^2} |\psi_1 |^2  + (1-\delta)\frac{16(1-p_2)Q^2u \TT}{(r^2+a^2)^3}  \frac{1}{r^2}|\psi_2|^2\\
&&+\frac{u \TT}{(r^2+a^2)^3}\Big(-O(a)\big(|\nab\psi_1|^2+r^{-2}|\nab_T\psi_1|^2+r^{-2}|\psi_1|^2+|\nab\psi_2|^2+r^{-2}|\nab_T\psi_2|^2+r^{-2}|\psi_2|^2\big)\Big)\\
&&-u O(aQ^2)r^{-6}\Big( \Re[ \psi_1 \c \left(  \DD \c  \ov{\psi_2}  \right)]+\Im[ \psi_1 \c \left(  \DD \c  \ov{\psi_2}  \right)]+r^{-1}|\psi_1|^2+r^{-1}|\psi_2|^2\Big)\\
&\geq&\delta \frac{u \TT}{(r^2+a^2)^3} \big( |\nab \psi_1|^2 +8Q^2 |\nab \psi_2 |^2 \big)\\
&&+(1-\delta)\frac{u \TT}{(r^2+a^2)^3}  \frac{1}{r^2} |\psi_1 |^2 - \frac{8p_1^2Q^2r^6}{p_2(1-\delta)|q|^{10}(r^2+a^2)}\frac{u}{ \TT}\Big( \frac{\TT}{r^2+a^2}+\frac{\De r}{2|q|^2} \Big)^2 |\psi_1|^2 \\
&&+ (1-\delta)\frac{16(1-p_2)Q^2u \TT}{(r^2+a^2)^3}  \frac{1}{r^2}|\psi_2|^2-\frac{32(1-p_1)^2Q^4r^6}{(1-\delta)|q|^{10}(r^2+a^2)}\frac{u}{ \TT} \Big( \frac{\TT}{r^2+a^2}+\frac{\De r}{2|q|^2} \Big)^2 |\psi_2|^2\\
&&+\frac{u \TT}{(r^2+a^2)^3}\Big(-O(a)\big(|\nab\psi_1|^2+r^{-2}|\nab_T\psi_1|^2+r^{-2}|\psi_1|^2+|\nab\psi_2|^2+r^{-2}|\nab_T\psi_2|^2+r^{-2}|\psi_2|^2\big)\Big)\\
&&-u O(aQ^2)r^{-6}\Big( \Re[ \psi_1 \c \left(  \DD \c  \ov{\psi_2}  \right)]+\Im[ \psi_1 \c \left(  \DD \c  \ov{\psi_2}  \right)]+r^{-1}|\psi_1|^2+r^{-1}|\psi_2|^2\Big)
\eeaa
where we used twice the inequality $ax^2-bxy\geq -\frac{b^2}{4a}y^2.$ By combining $\mbox{I}+\mbox{II}$
 we obtain the stated expressions.
\end{proof}

\subsubsection{Choice of functions}

We now define the function  $u$ as
\begin{eqnarray}\label{eq:definition-u}
    u=\Big(1-\frac{Q^2}{M^2}\Big) (r^2-r_P^2)+\frac{Q^2}{M^2}\frac{r^3-r_P^3}{r},
\end{eqnarray}  
where recall that $r_P$ is the largest root of the trapping polynomial $\TT= r^3-3Mr^2 + ( a^2+2Q^2)r+Ma^2$.
As a consequence of \eqref{eq:FF-wred-w}, the function $w_Y$ is given by 
\begin{eqnarray}\label{eq:definition-w}
   w_Y=\frac{\De}{(r^2+a^2)^2}\pr_r u= \frac{\De}{(r^2+a^2)^2}\left( 2r+\frac{Q^2}{M^2}\frac{r_P^3}{r^2}\right).
\end{eqnarray}

We have the following.

\begin{proposition}
With the choice of $u$ in \eqref{eq:definition-u}, for $|Q|<M$ and for $\delta$ and $\frac{|a|}{M}$ sufficiently small such that $Q^2+a^2<M^2$, we have
\begin{itemize}
    \item $\partial_r\left( \frac{   u}{r^2+a^2}  \right)>0$,
    \item if $p_1=\frac 2 5$ and $p_2=\frac 1 2$, then $A_1, A_2>0$ for $r\geq r_+$.
\end{itemize}
\end{proposition} 
\begin{proof}
    By continuity it suffices to show that the two conditions hold for $\delta=0$ and $a=0$.
The first condition for $a=0$ reduces to
    \begin{align*}
        \partial_r\big( \frac{   u}{r^2}  \big)&= \Big(1-\frac{Q^2}{M^2}\Big) \frac{2r_P^2}{r^3}+\frac{Q^2}{M^2}\frac{3r_P^3}{r^4},
    \end{align*}
    which is clearly positive for $|Q|<M$.

    The expressions for $A_1, A_2$ in \eqref{eq:definition-A1}-\eqref{eq:definition-A2} for $a=\delta=0$ reduce to
\begin{align*}
     A_1&= \frac{u}{r^5} \big( 1-\frac{3M}{r}+\frac{2Q^2}{r^2}  \big)   - \frac{8p_1^2Q^2}{p_2r^7}\frac{u}{  1-\frac{3M}{r}+\frac{2Q^2}{r^2}}\Big(1-\frac{3M}{r}+\frac{2Q^2}{r^2}+\frac 1 2  \Up  \Big)^2 \\
     &-  \frac 1 4 r^{-2} \pr_r\big(\De \pr_r w_Y \big)-\frac 1 2\frac{u}{r^2}  \pr_r \big(\Up  V_1\big), \\
     A_2&= \frac{2(1-p_2)u}{r^5} \big( 1-\frac{3M}{r}+\frac{2Q^2}{r^2}  \big) -\frac{4(1-p_1)^2Q^2}{r^7}\frac{u}{  1-\frac{3M}{r}+\frac{2Q^2}{r^2}}\Big(1-\frac{3M}{r}+\frac{2Q^2}{r^2}+\frac 1 2  \Up  \Big)^2 \\
     &-  \frac 1 4 r^{-2} \pr_r\big(\De \pr_r w_Y \big)-\frac 1 2 \frac{u}{r^2}  \pr_r \big(\Up  V_2\big),
 \end{align*}
 where $\Up=\frac{\De}{r^2}=1-\frac{2M}{r}+\frac{Q^2}{r^2}$.
 From Lemma \ref{lemma:positivity-RN}, we deduce that if $p_1=\frac 2 5$ and $p_2=\frac 1 2$ then $A_1$ and $A_2$ are positive for $|Q|<M$ and  $r\geq r_+$.
\end{proof}

 As a consequence of the above, for $|Q|<M$ and $\frac{|a|}{M}$ sufficiently small, there exists a sufficiently small universal constant $c_0$ such that 
  \beaa
\mbox{Main}&\geq& c_0 \Big[  r^{-3} \big(  |\nab_{\pr_{r_*}}\psi_1|^2 + |\nab_{\pr_{r_*}}\psi_2|^2+|\psi_1|^2+|\psi_2|^2 \big)+\big(1-\frac{r_P}{r} \big)^2  \Big(r^{-1} |\nab \psi_1|^2+ r^{-1} |\nab \psi_2|^2  \Big)\Big]\\
  &&-O(a)\frac 1 r \left(1-\frac{r_P}{r} \right)^2\big(|\nab\psi_1|^2+r^{-2}|\nab_T\psi_1|^2+|\nab\psi_2|^2+r^{-2}|\nab_T\psi_2|^2\big)-O(a)r^{-3} (|\psi_1|^2+|\psi_2|^2)\\
    &&-\D_\mu \Re\Big[\frac{4Q^2 q^3}{|q|^5} \big( \psi_1 \c \big(\nabla_Y\ov{\psi_2} +\frac 1 2   w_Y \ov{\psi_2} \big)\big)^\mu -\frac{4Q^2z q^3}{|q|^5}u\psi_1 \c  (\DD \c\ov{\psi_2}) \pr_r^\mu \Big]
\eeaa
where we distributed the trapping factor in the angular derivative in the terms $\Re[ \psi_1 \c \left(  \DD \c  \ov{\psi_2}  \right)]+\Im[ \psi_1 \c \left(  \DD \c  \ov{\psi_2}  \right)]$ through Cauchy-Schwarz. 

Upon combining the above with the standard use of a Lagrangian, given by $(Y, w_Y)+(0, \delta_Tw_T)$ for $w_T=- \frac{4 M \De \TT^2}{r^2 (r^2+a^2)^4}$ and some small $\delta_T>0$, we can add control of the trapped time derivative, so that the second line can be absorbed by the first one for $\frac{|a|}{M}$ sufficiently small, and can be bounded from below by $\Mor^{ax}[\psi_1, \psi_2](\tau_1, \tau_2)$.

On the other hand, the $\mbox{O(a) terms}$ in \eqref{eq:definition-Oaterms} satisfy
\beaa
|\mbox{O(a) terms}|&\les& O(a)\frac 1 r \left(1-\frac{r_P}{r} \right)^2\big(|\nab\psi_1|^2+r^{-2}|\nab_T\psi_1|^2+|\nab\psi_2|^2+r^{-2}|\nab_T\psi_2|^2\big)+O(a)r^{-3} (|\psi_1|^2+|\psi_2|^2)\\
&&+\D_\mu \Im \Big[  \frac{a\cos\th}{|q|^2} (\partial_r)^\mu   zu (\ov{\psi_1}\c\nab_T \psi_1+16Q^2\ov{\psi_2}\c\nab_T \psi_2 )\Big] \\
&&+\int_{\MM(0, \tau)} \Big( (|\nab_{\pr_{r_*}}\psi_1|+r^{-1}|\psi_1|)|N_1| + Q^2(|\nab_{\pr_{r_*}}\psi_2|+r^{-1}|\psi_2|)|N_2|\Big).
\eeaa

Applying the divergence theorem to $\D^\mu \PP_\mu^{( Y, w_Y+\delta_T w_T)}[\psi_1, \psi_2]$ we deduce 
\beaa
  \Mor^{ax}[\psi_1, \psi_2](0, \tau)  \les &&\int_{\pr\MM(0, \tau)}|M(\psi_1, \psi_2)|\\
&& +\int_{\MM(0, \tau)}\Big( | \nab_{\pr_{r_*}} \psi_1 | + r^{-1}|\psi_1| \Big)    |  N|+\Big( | \nab_{\pr_{r_*}} \psi_2 | + r^{-1}|\psi_2| \Big)    |  N_2|,
\eeaa
where $M(\psi_1,\psi_2)$
is a quadratic expression in $\psi_1$ and $\psi_2$ and their first derivatives which can easily be bounded by the supremum of the energy $\sup_{[0, \tau]}E[\psi_1,\psi_2](\tau)$.
Finally, by combining the above with the energy estimates in \eqref{eq:energy-estimates-conditional} multiplied by a large constant (so that the energies can absorb the boundary terms $\int_{\pr\MM(0, \tau)}|M(\psi_1, \psi_2)|$), we obtain for $\frac{|a|}{M}$ sufficiently small
\bea\label{eq:energy-morawetz-with-lot}
\begin{split}
E[\psi_1, \psi_2](\tau)+\Mor^{ax}[\psi_1, \psi_2](0, \tau)   &\les  E[\psi_1, \psi_2](0)+\mathcal{N}[\psi_1, \psi_2, N_1, N_2] (0, \tau),
 \end{split}
 \eea
 where
\bea\label{eq:definition-NN-psi1psi2N1N2}
\begin{split}
\mathcal{N}[\psi_1, \psi_2, N_1, N_2] (0, \tau)&:= \int_{\MM(0, \tau)}\Big( | \nab_{\pr_{r_*}} \psi_1 | + r^{-1}|\psi_1| \Big)    |  N_1|+\Big( | \nab_{\pr_{r_*}} \psi_2 | + r^{-1}|\psi_2| \Big)    |  N_2|\\
&+\int_{\MM(0, \tau)}\Big( |N_1|^2+Q^2 |N_2|^2\Big)\\
&+\left|\int_{\MM(0, \tau)}\Re\Big( \nabla_{T}\ov{\psi_1} \c N_1+8Q^2  \nabla_{T}\ov{\psi_2}  \c N_2\Big)\right|
\end{split}
\eea

\subsection{Control of the lower order terms for the gRW system}\label{sec:lot}

We are left to control the lower order terms $\mathcal{N}[\psi_1, \psi_2, N_1, N_2] (0, \tau)$ on the right hand side of \eqref{eq:energy-morawetz-with-lot}. This control was obtained in two proposition in \cite{Giorgi9} that we recall here. Notice that the two proposition did not have any restriction on the charge parameter.

\begin{proposition}[Proposition 3.5 in \cite{Giorgi9}]\label{lemma:crucial1} Let $\pf, \qf$ be solutions of the gRW system \eqref{final-eq-1}-\eqref{final-eq-2}. Then, the following holds true:
 \bea
\NN[\pf, \qf, N_1, N_2](0, \tau)
  \les |a|\big(  E[\pf, \qf, \Bfr, \Ffr, A, \Xfr] ( \tau)+B[\pf, \qf, \Bfr, \Ffr, A, \Xfr] (0, \tau)\big),
  \eea   
  where $N_1=  O(a^2 r^{-4}) \psi_1 + L_1$ and $N_2= O(a^2 r^{-4}) \psi_2 + L_2$, with $L_1$ and $L_2$ given by Theorem \ref{main-theorem-RW}.
\end{proposition}

\begin{proposition}[Proposition 3.6 in \cite{Giorgi9}]\label{lemma:crucial2} Let $\Bfr, \Ffr, A, \Xfr$ be solutions of the Teukolsky system and let $\pf, \qf$ defined as their Chandrasekhar transformation. For $|a| \ll M$, the following transport estimate holds true\footnote{Observe that Proposition 3.6 in \cite{Giorgi9} was stated with the Morawetz bulk $\Mor[\psi_1, \psi_2](\tau_1, \tau_2)$, which is controlled by the axially symmetric one $\mbox{Mor}^{ax}[\psi_1, \psi_2](\tau_1, \tau_2)$, hence the statement.}:
 \bea
 \lab{eq:transportA}
 E[\Bfr, \Ffr, A, \Xfr] ( \tau)+B[\Bfr, \Ffr, A, \Xfr] (0, \tau)&\les&   E[\pf, \qf](\tau) +\mbox{Mor}^{ax}[\pf, \qf](0,\tau)  + E[\Bfr, \Ffr, A, \Xfr] (0).
 \eea
\end{proposition}

Combining the two above proposition with \eqref{eq:energy-morawetz-with-lot} we finally deduce for $\frac{|a|}{M}$ sufficiently small
\bea
E[\pf, \qf](\tau)+\Mor^{ax}[\pf, \qf](0, \tau) &\les& E[\pf, \qf, \Bfr, \Ffr, A, \Xfr] (0),
\eea
which combined once again with \eqref{eq:transportA} gives 
 \beaa
E[\pf, \qf, \Bfr, \Ffr, A, \Xfr](\tau)+B[\pf, \qf, \Bfr, \Ffr, A, \Xfr] (0, \tau) \les  E[\pf, \qf, \Bfr, \Ffr, A, \Xfr](0),
       \eeaa
       therefore proving Theorem 
\ref{Thm:Nondegenerate-Morawetz}.

 \appendix

\section{The case of Reissner-Nordstr\"om}

In the case of $a=0$, the gRW system given by Theorem \ref{main-theorem-RW} reduces to the following coupled system of Regge-Wheeler equations:
\begin{eqnarray}
    \begin{split}
 \squared_1\pf -V_{1}  \pf &=&\frac{4Q^2}{r^2}   \ov{\DD} \c  \qf  \label{final-eq-1-RN}\\
\squared_2\qf -V_{2}  \qf &=&-  \frac{1}{2r^2}  \DD \hot  \pf   \label{final-eq-2-RN}        
    \end{split}
\end{eqnarray}
 where
  \begin{eqnarray*}
 V_{1}=\frac{1}{r^2}\big(1-\frac{2M}{r}+\frac{6Q^2}{r^2} \big), \qquad V_{2}=\frac{4}{r^2}\big(1-\frac{2M}{r}+\frac{3Q^2}{2r^2} \big).    
  \end{eqnarray*}

The system in Reissner-Nordstr\"om was first obtained in \cite{Giorgi7a}, where energy-Morawetz estimates were obtained for $|Q|<M$. Here, the choice of multiplier in Theorem \ref{Thm:Nondegenerate-Morawetz}
 gives an alternative proof of the energy-Morawetz estimates obtained in \cite{Giorgi7a}, which in a static spacetime are valid without restriction to axial symmetry.  We have

 \begin{theorem}
\lab{Thm:Nondegenerate-Morawetz-RN}
Let $\pf, \qf$ be solutions the Regge-Wheeler system \eqref{final-eq-1-RN} in Reissner-Nordstr\"om spacetime. Then, for $|Q| <M$  the following energy boundedness and integrated local energy decay estimates hold true:
 \bea\label{eq:final-estimate-theorem-RN}
E[\pf, \qf](\tau)+\Mor^{ax}[\pf, \qf] (0, \tau) \les  E[\pf, \qf](0).
       \eea
\end{theorem}

 The above theorem is obtained by combining the energy estimates in Proposition \ref{prop:energy-estimates-conditional} for $a=0$ and the coercivity of the spacetime energy in  Morawetz estimates through the choice of functions given in \eqref{eq:definition-u}-\eqref{eq:definition-w}.
 Notice that the choice of functions that we have here, unlike the ones in \cite{Giorgi7a}, has the advantage of not using the spherical harmonics decomposition in the proof of coercivity and can therefore be adapted to the general case of Kerr-Newman. 
 Observe that, making a connection with the notations in \cite{Giorgi7a}, we have $z=r^{-2} \Upsilon$ for $\Up=1-\frac{2M}{r}+\frac{Q^2}{r^2}$ and $f=r^{-2} u$.

For $a=0$, since $\mathscr{N}_{first}^{( Y, w_Y)}[\psi_1,\psi_2], \mathscr{N}_{lot}^{(Y, w_Y)}[\psi_1,\psi_2], \mathscr{R}^{(Y)}[\psi_1, \psi_2]=0 $ the Morawetz divergence in \eqref{eq:divergence-theorem-identity-Y} from \eqref{eq:expression-I-II} reduces to
\beaa
\D^\mu \PP_\mu^{( Y, w_Y)}[\psi_1, \psi_2]&=& \EE^{(Y, w_Y)}[\psi_1, \psi_2] +\mathscr{N}_{coupl}^{( Y, w_Y)}[\psi_1,\psi_2]\\
&=&\Up^2 \partial_r\big( \frac{   u}{r^2}  \big) \big(  |\nab_r\psi_1|^2 + 8Q^2|\nab_r\psi_2|^2 \big)+A_1 |\psi_1|^2+8Q^2 A_2|\psi_2|^2 \\
  &&-\D_\mu \Re\Big[\frac{4Q^2}{r^2} \big( \psi_1 \c \big(\nabla_Y\ov{\psi_2} +\frac 1 2   w_Y \ov{\psi_2} \big)\big)^\mu -\frac{4Q^2\Up}{r^4}u\psi_1 \c  (\DD \c\ov{\psi_2}) \pr_r^\mu \Big],
\eeaa
 where 
 \begin{align*}
     A_1&= \frac{u}{r^5} \big( 1-\frac{3M}{r}+\frac{2Q^2}{r^2}  \big)   - \frac{8p_1^2Q^2}{p_2r^7}\frac{u}{  1-\frac{3M}{r}+\frac{2Q^2}{r^2}}\Big(1-\frac{3M}{r}+\frac{2Q^2}{r^2}+\frac 1 2  \Up  \Big)^2 \\
     &-  \frac 1 4 r^{-2} \pr_r\big(\De \pr_r w_Y \big)-\frac 1 2\frac{u}{r^2}  \pr_r \big(\Up  V_1\big), \\
     A_2&= \frac{2(1-p_2)u}{r^5} \big( 1-\frac{3M}{r}+\frac{2Q^2}{r^2}  \big) -\frac{4(1-p_1)^2Q^2}{r^7}\frac{u}{  1-\frac{3M}{r}+\frac{2Q^2}{r^2}}\Big(1-\frac{3M}{r}+\frac{2Q^2}{r^2}+\frac 1 2  \Up  \Big)^2 \\
     &-  \frac 1 4 r^{-2} \pr_r\big(\De \pr_r w_Y \big)-\frac 1 2 \frac{u}{r^2}  \pr_r \big(\Up  V_2\big).
 \end{align*}

We have the following.
\begin{lemma}\label{lemma:positivity-RN}
   If $p_1=\frac 2 5$ and $p_2=\frac 1 2$, then the expressions $A_1, A_2$ are positive for $r\geq r_+$ in Reissner-Nordstr\"om for $|Q|<M$.
\end{lemma}

\begin{proof}
    We denote $\tilde r_P= \frac{3M-\sqrt{9M^2-8Q^2}}{2}<r_H=M+\sqrt{M^2-Q^2}$, which is another root of $r^2-3Mr+Q^2=0$, other than $r_P=\frac{3M+\sqrt{9M^2-8Q^2}}{2}$. Through direct computations, we can write down
  \begin{eqnarray}
    \begin{split}    
&\frac{u}{r^5} \big( 1-\frac{3M}{r}+\frac{2Q^2}{r^2}  \big) =\frac{1}{r^3}  \left(1-\frac{\tilde r_P}{r} \right) \left(1-\frac{r_P}{r}\right)^2 \left(1+ \frac{ r_P}{r}+\frac{Q^2}{M^2} \frac{ r_P^2}{r^2}  \right),\\     
&\frac{Q^2}{r^7}  \frac{u}{  1-\frac{3M}{r}+\frac{2Q^2}{r^2}}\Big(1-\frac{3M}{r}+\frac{2Q^2}{r^2}+\frac 1 2  \Up  \Big)^2 = \frac{1}{r^3} \frac{Q^2}{r^2} \left(1+ \frac{ r_P}{r}+\frac{Q^2}{M^2} \frac{ r_P^2}{r^2}  \right)\frac{r}{r-\tilde r_P}  \left(\frac{3}{2}-\frac{4M}{r}+ \frac{5 Q^2}{2 r^2} \right)^2,\\
&-  \frac 1 4 r^{-2} \pr_r\big(\De \pr_r w_Y \big)= \frac{1}{r^3}
\left(\frac{3M}{r}-\frac{8M^2}{r^2}-\frac{4Q^2}{r^2}+\frac{15Q^2 M}{r^3}-\frac{6Q^4}{r^4} \right)\\
&\qquad\qquad\qquad\qquad\qquad+\frac{1}{r^3} \frac{Q^2}{M^2}\frac{r_P^3}{r^3} \left(-  3+\frac{18M}{r}  -\frac{25 M^2}{r^2} -\frac{25 Q^2}{2r^2}+\frac{33 Q^2 M}{r^3} -\frac{21 Q^4}{2r^4}  \right), \\
&-\frac 1 2 \frac{u}{r^2}  \pr_r \big(\Up  V_1\big)   = \frac{1}{r^3} \left(1-\frac{r_P}{r} \right)\left(1+ \frac{ r_P}{r}+\frac{Q^2}{M^2} \frac{ r_P^2}{r^2}  \right)  \left(1-\frac{6 M}{r}+\frac{8 M}{r^2}+\frac{14Q^2}{r^2}-\frac{35Q^2M}{r^3}+\frac{18Q^4}{r^4}\right),      \\
&-\frac 1 2 \frac{u}{r^2}   \pr_r \big(\Up  V_2\big)  =\frac{2}{r^3} \left(1-\frac{r_P}{r} \right)\left(1+ \frac{ r_P}{r}+\frac{Q^2}{M^2} \frac{ r_P^2}{r^2}  \right) \left(2-\frac{12 M}{r}+\frac{16 M^2}{r^2}+\frac{10Q^2}{r^2}-\frac{25Q^2 M}{r^3}+\frac{9Q^4}{r^4}\right).    
    \end{split}
    \end{eqnarray}   

With these expressions, we can express $r^3A_1$, $r^3A_2$ when $p_1=\frac 2 5$ and $p_2=\frac 1 2$ as follows, using dimensionless variables $x=r/M$ and $\lambda=|Q|/M$, $x_p=r_P/M=\frac{3}{2}+\sqrt{\frac{9}{4}-2\lambda^2}$, $\tilde{x}_p= \tilde {r}_p/M=\frac{3}{2}-\sqrt{\frac{9}{4}-2\lambda^2}$.   We have
\begin{equation}\label{eq:r3A1}
    \begin{split}
       r^3 A_1=&\left(1-\frac{\tilde x_p}{x} \right) \left(1-\frac{x_p}{x}\right)^2 \left(1+ \frac{ x_p}{x}+\frac{\lambda^2 x_p^2}{x^2}  \right)\\
       &-\frac{64}{25}\frac{\lambda^2}{x^2}\left(1+ \frac{ x_p}{x}+\frac{\lambda^2 x_p^2}{x^2}  \right)\frac{x }{x-\tilde{x}_p }\left(\frac{9}{4}-\frac{12}{x}+\frac{16}{x^2} +\frac{15\lambda^2}{2x^2}-\frac{20\lambda^2}{x^3}+\frac{25\lambda^4}{4x^4}\right)\\
        &+\left(\frac{3}{x}-\frac{8}{x^2}-\frac{4\lambda^2}{x^2}+\frac{15\lambda^2}{x^3}-\frac{6\lambda^4}{x^4} \right)+\lambda^2\frac{x_p^3}{x^3} \left(-  3+\frac{18}{x}  -\frac{25 }{x^2} -\frac{25 \lambda^2}{2x^2}+\frac{33 \lambda^2}{x^3} -\frac{21 \lambda^4}{2x^4}  \right)   \\
        &+ \left(1-\frac{x_p}{x} \right)\left(1+ \frac{x_p}{x} + \frac{\lambda^2 x_p^2}{x^2} \right)    \left(1-\frac{6}{x}+\frac{8}{x^2}+\frac{14\lambda^2}{x^2}-\frac{35\lambda^2}{x^3}+\frac{18\lambda^4}{x^4}\right) ,   
    \end{split}
\end{equation}
and
\begin{equation}\label{eq:r3A2}
    \begin{split}
      r^3  A_2 =& \left(1-\frac{\tilde x_p}{x} \right) \left(1-\frac{x_p}{x}\right)^2 \left(1+ \frac{ x_p}{x}+\frac{\lambda^2 x_p^2}{x^2}  \right)\\
      &-\frac{36}{25}\frac{\lambda^2}{x^2}\left(1+ \frac{ x_p}{x}+\frac{\lambda^2 x_p^2}{x^2}  \right)\frac{x }{x-\tilde{x}_p }\left(\frac{9}{4}-\frac{12}{x}+\frac{16}{x^2} +\frac{15\lambda^2}{2x^2}-\frac{20\lambda^2}{x^3}+\frac{25\lambda^4}{4x^4}\right)\\
        &+\left(\frac{3}{x}-\frac{8}{x^2}-\frac{4\lambda^2}{x^2}+\frac{15\lambda^2}{x^3}-\frac{6\lambda^4}{x^4} \right)+\lambda^2\frac{x_p^3}{x^3} \left(-  3+\frac{18}{x}  -\frac{25 }{x^2} -\frac{25 \lambda^2}{2x^2}+\frac{33 \lambda^2}{x^3} -\frac{21 \lambda^4}{2x^4}  \right)   \\
        &+2\left(1-\frac{x_p}{x} \right)\left(1+ \frac{x_p}{x} + \frac{\lambda^2 x_p^2}{x^2} \right) \left(2-\frac{12}{x}+\frac{16}{x^2}+\frac{10\lambda^2}{x^2}-\frac{25\lambda^2}{x^3}+\frac{9\lambda^4}{x^4}\right) .
    \end{split}
\end{equation}

 Notice that in exterior of sub-extremal case, $r_+>\tilde{r}_p$ always holds true. Therefore, the positivity of \eqref{eq:r3A1} and \eqref{eq:r3A2} for $r\geq r_+$ and $|Q|<M$ is equivalent to the positivity of $\left(1-\frac{\tilde x_p}{x} \right)r^3 A_1$ and  $\left(1-\frac{\tilde x_p}{x} \right)r^3 A_2$ for $r\geq r_+$ and $|Q|<M$. 
    
 Making use of the fact that $\tilde x_p=3-x_p$ and $\lambda^2=\frac{3x_p-x_p^2}{2}$ to replace all the $\tilde x_p$  and $\lambda$ in the expressions, then $\left(1-\frac{\tilde x_p}{x} \right)r^3 A_1$ and  $\left(1-\frac{\tilde x_p}{x} \right)r^3 A_2$ can be expanded as polynomial expressions\footnote{A Mathematica notebook containing the verifications of these expressions is included as an ancillary arXiv file.} of $x^{-1},x_p$:
\begin{equation}\label{eq:r3A1-poly}
    \begin{split}
\left(1-\frac{\tilde x_p}{x} \right) r^3 A_1 
=&2- \frac{12}{x} + \frac{2 x_p}{x} + \frac{18}{x^2}  + \frac{84 x_p}{
 25 x^2}  - \frac{128 x_p^2}{25 x^2}  + \frac{3 x_p^3}{x^2} - \frac{x_p^4}{x^2} 
\\
&- \frac{948 x_p}{25 x^3}+ \frac{37 x_p^2}{x^3} - \frac{1381 x_p^3}{50 x^3} + \frac{3 x_p^4}{x^3}+ \frac{3 x_p^5}{2 x^3}\\
 &+ \frac{714 x_p}{25 x^4} - \frac{1116 x_p^2}{25 x^4} + \frac{2147 x_p^3}{50 x^4}  + \frac{1531 x_p^4}{25 x^4} - \frac{984 x_p^5}{25 x^4} + \frac{253 x_p^6}{50 x^4}\\
 &- \frac{81 x_p^2}{25 x^5}  + \frac{3099 x_p^3}{50 x^5} - \frac{31553 x_p^4}{100 x^5} + \frac{29999 x_p^5}{200 x^5} - \frac{391 x_p^6}{50 x^5} - \frac{25 x_p^7}{8 x^5} \\
 &- \frac{963 x_p^3}{10 x^6} + \frac{33429 x_p^4}{100 x^6} + \frac{14363 x_p^5}{200 x^6}  - \frac{41563 x_p^6}{200 x^6} + \frac{2841 x_p^7}{40 x^6} - \frac{279 x_p^8}{40 x^6} \\
 &+ \frac{135 x_p^4}{2 x^7}  - \frac{10719 x_p^5}{20 x^7} + \frac{36441 x_p^6}{80 x^7} - \frac{9857 x_p^7}{80 x^7} + \frac{443 x_p^8}{80 x^7}  + \frac{21 x_p^9}{16 x^7} \\
 &+ \frac{3321 x_p^6}{16 x^8}   - \frac{
 1107 x_p^7}{4 x^8}  + \frac{1107 x_p^8}{8 x^8}  - \frac{123 x_p^9}{4 x^8}+ \frac{
 41 x_p^{10}}{16 x^8}
    \end{split}
\end{equation}
and
\begin{equation}\label{eq:r3A2-poly}
    \begin{split}
\left(1-\frac{\tilde x_p}{x} \right) r^3 A_2
=& 5 -\frac{39}{x} +\frac{5 x_p}{x}  + \frac{96}{x^2} -\frac{ 93 x_p}{50x^2}-\frac{619 x_p^2 }{50 x^2}  + \frac{15 x_p^3}{2 x^2} - \frac{5 x_p^4}{2 x^2} \\
&- \frac{72}{x^3} - \frac{4179 x_p}{50 x^3}  + \frac{100 x_p^2}{x^3} - \frac{3769 x_p^3}{50 x^3}+ \frac{33 x_p^4}{2x^3} + \frac{3 x_p^5}{2 x^3}\\
 &+ \frac{6147 x_p}{50 x^4} - \frac{8893 x_p^2}{50 x^4} + \frac{4264 x_p^3}{25 x^4} + \frac{6001 x_p^4}{100 x^4}- \frac{3057 x_p^5}{50 x^4} + \frac{869 x_p^6}{100 x^4} \\
 &+ \frac{1131 x_p^2}{25 x^5} + \frac{1151 x_p^3}{50 x^5} - \frac{12718 x_p^4}{25 x^5} + \frac{55751 x_p^5}{200 x^5} - \frac{2843 x_p^6}{100 x^5} - \frac{25 x_p^7}{8 x^5} \\
 &- \frac{7623 x_p^3}{40 x^6} + \frac{127767 x_p^4}{200 x^6} - \frac{1547 x_p^5}{50 x^6} - \frac{24981 x_p^6}{100 x^6} + \frac{3759 x_p^7}{40 x^6}- \frac{381 x_p^8}{40 x^6} \\
 &+ \frac{729 x_p^4}{8 x^7}     - \frac{29457 x_p^5}{40 x^7} + \frac{51219 x_p^6}{80 x^7}- \frac{14643 x_p^7}{80 x^7} + \frac{967 x_p^8}{80 x^7}+ \frac{21 x_p^9}{16 x^7} \\
 &+ \frac{243 x_p^6}{x^8}   - \frac{324 x_p^7}{x^8}  + \frac{162 x_p^8}{x^8}  - \frac{36 x_p^9}{x^8} + \frac{3 x_p^{10}}{x^8}.
    \end{split}
\end{equation}

It is straightforward to verify that 
\[\Bigg\{x\geq 1+\sqrt{1-\lambda^2},\quad 0\leq \lambda <1   \Bigg\}= \Bigg\{x\geq 1+\sqrt{\frac{(x_p-1)(x_p-2)}{2}}, \quad 2< x_p \leq 3 \Bigg\}\]

Using standard calculus techniques it is possible to verify that\footnote{We refer the reader to the Mathematica notebook for the verification.} 
 that \eqref{eq:r3A1-poly} and \eqref{eq:r3A2-poly} are positive for $x\in \left[1+\sqrt{\frac{(x_p-1)(x_p-2)}{2}},\infty  \right)$ and $x_p\in (2,3]$.   This means $r^3A_1>0, r^3 A_2>0$ if $r\geq r_+$ and $|Q|<M$ as desired.

\end{proof}

\begin{remark} The expressions \eqref{eq:r3A1} and \eqref{eq:r3A2} depend on both $x$ and $\lambda$ in the exterior region of full sub-extremal case. We can examine the positivity of them in Schwarzschild ($Q=0$) and extremal Reissner-Nordstr\"om ($Q=M$) cases explicitly.
    
\begin{enumerate}
    \item In the case of Schwarzschild, we have $\lambda=0$, $x_p=3$ and $\tilde{x}_p=0$, so \eqref{eq:r3A1} and \eqref{eq:r3A2} reduce to
    \begin{equation}
 \begin{split}
      r^3  A_1^{(\text{Schw})}=& \left(1-\frac{3}{x}\right)^2 \left(1+\frac{3}{x}\right)
        +\left ( \frac{3}{x} -\frac{8}{x^2}\right) + \left(1-\frac{3}{x} \right)\left(1+ \frac{3}{x}  \right)    \left(1-\frac{6}{x}+\frac{8}{x^2}\right)\\
        =&2 -\frac{6}{x} -\frac{18}{x^2} +\frac{81}{x^3}-\frac{72}{x^4},
    \end{split}        
    \end{equation}
    and
   \begin{equation}
 \begin{split}
     r^3   A_2^{(\text{Schw})}=& \left(1-\frac{3}{x}\right)^2 \left(1+\frac{3}{x}\right)
        +\left( \frac{3}{x} -\frac{8}{x^2} \right)+ 2\left(1-\frac{3}{x} \right)\left(1+ \frac{3}{x}  \right)    \left(2-\frac{12}{x}+\frac{16}{x^2}\right)\\
        =&5 -\frac{24}{x} -\frac{21}{x^2} +\frac{243}{x^3}-\frac{288}{x^4}.
    \end{split}        
    \end{equation}  

  These are single variable functions and using standard calculus techniques it is  straightforward to verify that $r^3 A_1^{(\text{Schw})}>0$ and $r^3 A_2^{(\text{Schw})}>0$
  for $x\geq 2$, i.e. $r\geq r_+^{(\text{Schw})}=2M$.
    \item In the case of extremal Reissner-Nordstr\"om, we have $\lambda=1$, $x_p=2$ and $\tilde{x}_p=1$, so \eqref{eq:r3A1} and \eqref{eq:r3A2} reduce to
    \begin{equation}
 \begin{split}
     r^3   A_1^{(\text{eRN})}=&\left(1-\frac{1}{x} \right) \left(1-\frac{2}{x}\right)^2 \left(1+ \frac{ 2}{x}+\frac{  2^2}{x^2}  \right)-\frac{64}{25}\frac{1}{x^2}\frac{x^2+2 x +   2^2}{x^2- x}\left(\frac{9}{4}-\frac{12}{x}+\frac{16}{x^2} +\frac{15 }{2x^2}-\frac{20 }{x^3}+\frac{25}{4x^4}\right)\\
        &+\left(\frac{3}{x}-\frac{4}{x^2}-\frac{8}{x^2}+\frac{15 }{x^3}-\frac{6 }{x^4} \right)+\frac{2^3}{x^3} \left(-  3 +\frac{18 }{x} -\frac{25 }{x^2} -\frac{25  }{2x^2} +\frac{33  }{x^3} -\frac{21 }{2x^4}  \right)   \\
        &+ \left(1-\frac{2}{x} \right)\left(1+ \frac{2}{x} + \frac{  2^2}{x^2} \right)    \left(1-\frac{6}{x}+\frac{8}{x^2}+\frac{14 }{x^2}-\frac{35 }{x^3}+\frac{18 }{x^4}\right)\\
        =& \frac{2 (x - 1) (25 x^6 - 50 x^5 + 28 x^4 - 554 x^3 + 2192 x^2 - 3390 x + 2050)}{25 x^7},
    \end{split}        
    \end{equation}
    and
   \begin{equation}
 \begin{split}
    r^3    A_2^{(\text{eRN})}=&\left(1-\frac{1}{x} \right) \left(1-\frac{2}{x}\right)^2 \left(1+ \frac{ 2}{x}+\frac{2^2}{x^2}  \right)-\frac{36}{25}\frac{1}{x^2}\frac{x^2+2 x +  2^2}{x^2-  x}\left(\frac{9}{4}-\frac{12}{x}+\frac{16}{x^2} +\frac{15 }{2x^2}-\frac{20 }{x^3}+\frac{25}{4x^4}\right)\\
        &+\left(\frac{3}{x}-\frac{4 }{x^2}-\frac{8}{x^2}+\frac{15 }{x^3}-\frac{6 }{x^4} \right)+\frac{2^3}{x^3} \left(-  3 +\frac{18}{x} -\frac{25 }{x^2} -\frac{25 }{2x^2} +\frac{33 }{x^3} -\frac{21 }{2x^4}  \right)   \\
        &+2\left(1-\frac{2}{x} \right)\left(1+ \frac{2}{x} + \frac{2^2}{x^2} \right) \left(2-\frac{12}{x}+\frac{16}{x^2}+\frac{10 }{x^2}-\frac{25}{x^3}+\frac{9 }{x^4}\right) \\
        =&\frac{(x - 1) (125 x^6 - 475 x^5 + 494 x^4 - 1792 x^3 + 7391 x^2 - 10270 x + 4800)}{25 x^7}
    \end{split}        
    \end{equation}  
    Once again, these are single variable functions and it is straightforward to verify that $r^3 A_1^{(\text{eRN})}>0$ and $r^3 A_2^{(\text{eRN})}>0$
    whenever $x>1$, i.e. $r>r_+^{(\text{eRN})}=M$. Here, the degeneracy along the event horizon, corresponding to the factor $(x-1)$, is an expected consequence of the Aretakis instability \cite{Aretakis11}.
\end{enumerate}

\end{remark}

\end{document}